%% file: ms.tex
\RequirePackage[l2tabu,orthodox]{nag}% nag when using outdated macros
\RequirePackage{amsmath,amssymb}
\documentclass[a4paper]{scrartcl}
\usepackage{fullpage}

\allowdisplaybreaks[1] % allow page breaks in {align} environments

\usepackage[utf8]{inputenc}

\usepackage{hyperref}

\renewcommand\path[1]{{\sffamily\small\detokenize{#1}}}
\usepackage{xcolor}
\hypersetup{
    colorlinks,
    linkcolor={blue!50!black},
    citecolor={blue!50!black},
    urlcolor={black}
}

\usepackage[amsmath,thmmarks,hyperref]{ntheorem}
\theoremseparator{.}
\newtheorem{theorem}{Theorem}
\newtheorem{lemma}[theorem]{Lemma}
\newtheorem{proposition}[theorem]{Proposition}
\newtheorem{corollary}[theorem]{Corollary}

\theoremstyle{plain}
\theorembodyfont{\upshape}

\newtheorem{remark}[theorem]{Remark}
\newtheorem{definition}[theorem]{Definition}
\newtheorem{fact}[theorem]{Fact}

\theoremheaderfont{\normalfont\itshape}
\newtheorem{claim}[theorem]{Claim}

\theoremsymbol{\ensuremath{_\blacksquare}}
\makeatletter
\theoremheaderfont{\normalfont\itshape}
\newtheoremstyle{MyNonumberplain}%
  {\item[\theorem@headerfont\hskip\labelsep ##1\theorem@separator]}%
  {\item[\theorem@headerfont\hskip\labelsep ##3\theorem@separator]}
\makeatother
\theoremstyle{MyNonumberplain}
\theorembodyfont{\upshape}
\newtheorem{proof}{Proof}
\theoremsymbol{\ensuremath{_\square}}
\newtheorem{claimproof}{Proof}

\usepackage{tikz}
\usetikzlibrary{decorations.pathreplacing,angles,quotes}
\usepackage{standalone}

\usepackage{enumitem}
\newenvironment{algor}[3]{%
\bigskip
\noindent{\sffamily\bfseries Algorithm #1} ({\itshape#2\/}) {\itshape #3}
\begin{description}[labelindent=0pt,labelwidth=1.7em,labelsep=0pt,leftmargin=!]%
\vspace{-.8ex}}{%
\end{description}\medskip}

\usepackage{mathtools}
\DeclarePairedDelimiter\paren{\lparen}{\rparen}
\DeclarePairedDelimiter\abs{\lvert}{\rvert}
\DeclarePairedDelimiter\set{\{}{\}}
\DeclarePairedDelimiterX\setc[2]{\{}{\}}{#1 \colon #2}
\DeclarePairedDelimiterX\parenc[2]{\lparen}{\rparen}{\,#1 \;\delimsize\vert\; #2\,}

\DeclarePairedDelimiterX\msetc[2]{\{\!\!\{}{\}\!\!\}}{#1 \delimsize: #2}

\newcommand{\cc}[1]{\ensuremath{\mathrm{#1}}}
\newcommand{\pp}[1]{\textup{#1}}
\newcommand{\op}[1]{\ensuremath{\operatorname{#1}}}

\renewcommand{\sp}{\#\cc{P}}

\newcommand{\W}{\cc{W[1]}}
\newcommand{\sw}{\cc{\#W[1]}}
\newcommand{\sharpP}{\cc{\#P}}

\newcommand{\poly}{\op{poly}}

\newcommand{\N}{\mathbb{N}}
\newcommand{\Z}{\mathbb{Z}}
\newcommand{\Q}{\mathbb{Q}}

\newcommand{\dotcup}{\mathbin{\dot\cup}}
\newcommand{\dotbigcup}{\mathbin{\dot\bigcup}}

\newcommand{\zo}{\set{0,1}}

\newcommand\restr[2]{{%
		\left.\kern-\nulldelimiterspace %
		#1 %
		\vphantom{\big|} %
		\right|_{#2} %
	}}

\newcommand{\classH}{\ensuremath{{\mathcal{H}}}}

\newcommand{\ch}{\ensuremath{\mathcal{H}}}

\newcommand{\Hom}{\ensuremath{{\mathrm{Hom}}}}
\newcommand{\eiHom}{\ensuremath{{\mathrm{EdgInj}}}}

\newcommand{\Sub}{\ensuremath{{\mathrm{Sub}}}}
\newcommand{\Emb}{\ensuremath{{\mathrm{Emb}}}}

\newcommand{\EC}{\mathsf{EC}}

\newcommand{\coloralloc}{\mathcal{K}}

\global\long\def\vc#1{\mathrm{vc}(#1)}
\global\long\def\partitions#1{\mathrm{Part}(#1)}
\global\long\def\contract#1#2{#1 / #2}

\global\long\def\hw{\mathrm{hw}}
\global\long\def\pCol{\pp{\#ColMatch}}
\global\long\def\ColSig{\mathrm{ColSig}}
\global\long\def\ColHolant{\mathrm{ColHolant}}
\global\long\def\Holant{\mathrm{Holant}}
\global\long\def\sigHW#1{\mathrm{hw_{#1}}}

\title{Counting edge-injective homomorphisms
  and matchings on restricted graph classes%
}

\usepackage{authblk}

\author[a]{Radu Curticapean}
\author[b]{Holger Dell}
\author[b,c]{Marc Roth}

\affil[a]{%
  Institute for Computer Science and Control,
  Hungarian Academy of Sciences (MTA SZTAKI)\\
  Hungary\\
  {radu.curticapean@gmail.com} 
}

\affil[b]{%
  Saarland University and Cluster of Excellence (MMCI), Saarbrücken, Germany%
  \\\{hdell,mroth\}@mmci.uni-saarland.de
}

\affil[c]{%
  Graduate School of Computer Science, Saarland University%
}

\date{January 18, 2018%
\thanks{Most of this work was done while the authors were visiting the Simons Institute for the Theory of Computing. Radu Curticapean is supported by ERC grants PARAMTIGHT (No.~280152) and SYSTEMATICGRAPH (No.~725978).}%
}

\begin{document}

\maketitle

\begin{abstract}
  We consider the \sw-hard problem of counting all matchings with exactly~$k$ edges in a given input graph~$G$;
  we prove that it remains \sw-hard on graphs~$G$ that are line graphs or bipartite graphs with degree $2$ on one side.

  In our proofs, we use that $k$-matchings in line graphs can be equivalently viewed as
  edge-injective homomorphisms from the disjoint union of $k$ length-$2$ paths into (arbitrary) host graphs.
  Here, a homomorphism from $H$ to $G$ is \emph{edge-injective} if it maps any two distinct
  edges of $H$ to distinct edges in $G$.
  We show that edge-injective homomorphisms from a pattern graph $H$ can be
  counted in polynomial time if $H$ has bounded vertex-cover number after
  removing isolated edges.
  For hereditary classes~$\mathcal{H}$ of pattern graphs, we complement this result: 
  If the graphs in~$\mathcal{H}$ have unbounded vertex-cover number even after deleting isolated edges,
  then counting edge-injective homomorphisms with patterns from $\mathcal{H}$ is \sw-hard.

  Our proofs rely on an edge-colored variant of Holant problems and a delicate
  interpolation argument; both may be of independent interest.
\end{abstract}

\section{Introduction}

Since Valiant's seminal $\sharpP$-hardness result for the permanent \cite{Valiant1979a}, 
various refinements of classical counting complexity have been studied, such
as approximate~\cite{Jerrum.Sinclair2004}, modular~\cite{Cai.Hemachandra1990}, and subexponential counting 
\cite{Dell.Husfeldt2014}, with additional
restrictions on the input classes \cite{Jerrum1987,Xia.Zhang2007}.
In this paper, we study counting problems through the lens of \emph{parameterized complexity}~\cite{Flum.Grohe2004}, 
where the input comes with a parameter $k\in\mathbb{N}$,
and we want to understand whether a problem admits a \emph{fixed-parameter tractable (FPT)} algorithm,
that is, one with running time $f(k)\cdot\poly(|x|)$ for some function~$f$.
The analogue of $\sharpP$ in this setting is the complexity class $\sw$,
for which counting cliques of size $k$ is a canonical complete problem.

In parameterized counting complexity, 
the problem of counting $k$-matchings plays an important role,
because it captures the complexity inherent to the counting version of the subgraph isomorphism problem.
Indeed if $H$ is a graph with a maximum matching of size~$\nu$, then we can count in time $n^{O(\nu)}$ all occurrences of~$H$ as a (not necessarily induced) subgraph of a given $n$-vertex graph~$G$.
There is strong evidence that the dependency on~$\nu$ is necessary:
For any class~$\mathcal{H}$ of graphs containing arbitrarily large matchings,
it is $\sw$-complete to count $H$-subgraphs~\cite{curticapean2014complexity}
even when $H$ is required to be from $\mathcal{H}$.
Furthermore, an $n^{o(\nu/\log \nu)}$ time algorithm for this restricted problem violates the exponential time hypothesis~\cite{DBLP:conf/stoc/CurticapeanDM17}.

In this paper, we proceed from the $\sw$-hardness result for counting $k$-matchings
in two ways: 
First, we strengthen the result by showing that counting $k$-matchings remains \sw-complete
even on natural \emph{restricted} graph classes, such as line graphs and bipartite graphs where one side has maximum degree
$2$.
As an instrument in our proofs, we introduce the notion of \emph{edge-injective
homomorphisms}, which interpolate between the classical notions of
homomorphisms and (subgraph) embeddings.
In the second part of the paper, we study the parameterized complexity of counting such edge-injective homomorphisms as a topic in itself.
This also relates to ``graph motif parameters'' \cite{DBLP:conf/stoc/CurticapeanDM17}, a recently introduced framework
for pattern counting problems
that was adapted from works by Lov\'asz~\cite{lovaszbook}.

\subsection{Counting matchings in restricted graph classes}

In non-parameterized counting complexity, restrictions of hard problems to planar and bounded-degree
graphs were studied extensively: We can count \emph{perfect}
matchings on planar graphs in polynomial time by the FKT method \cite{Temperley.Fisher1961,Kasteleyn1967}, and several dichotomies show which $\sharpP$-hard counting versions of constraint satisfaction problems become polynomial-time solvable on planar graphs \cite{DBLP:journals/siamcomp/CaiLX17,DBLP:journals/corr/CaiF16}.

Counting (not necessarily perfect) matchings has been studied by many authors~\cite{Jerrum1987,Dagum.Luby1992,Vadh2001}, culminating in the work of Xia~et~al.~\cite{Xia.Zhang2007} who showed that the problem remains
$\sharpP$-hard even on planar bipartite graphs whose left and right side have 
maximum degree $2$ and $3$, respectively.
In the parameterized setting, counting $k$-matchings is FPT in planar or
bounded-degree graphs~\cite{Frick2004}, which rules out a parameterized
analogue of the hardness result by Xia~et~al.~\cite{Xia.Zhang2007}.
On the other hand, counting $k$-matchings remains $\sw$-complete on bipartite
graphs, and this result was essential for a subsequent reduction to the
general subgraph counting problem~\cite{curticapean2014complexity}.
This reduction was recently superseded by~\cite{DBLP:conf/stoc/CurticapeanDM17}.

\subsubsection{Restricted bipartite graphs of high girth}

In~\cite{curticapean2014complexity},
the $\sw$-completeness of counting $k$-matchings in bipartite graphs was first shown for an edge-colorful variant, 
which was then reduced to the uncolored version via inclusion--exclusion.
In the edge-colorful variant, the edges of the bipartite graph are (not
necessarily properly) colored with $k$ colors and we wish to count
$k$-matchings that pick exactly one edge from each color.

In this paper, we strengthen the $\sw$-hardness
result for counting edge-colorful $k$-matchings in bipartite graphs~$G$ by showing that the problem remains hard when we restrict one side of $G$ to have maximum degree two.
We may further assume any constant lower bound on the \emph{girth} of $G$, that is, the length of the shortest cycle in~$G$.
For counting (edge-colorful) $k$-matchings, 
it was known before~\cite{curticapean2014complexity} that an algorithm with 
running time $f(k)\cdot n^{o(k/ \log k )}$ for any computable function~$f$ would 
refute the counting exponential-time hypothesis $\cc{\#ETH}$ \cite{Dell.Husfeldt2014}.
That is, if such an algorithm existed, we could 
count satisfying assignments to $n$-variable $3$-CNF formulas in time 
$\exp\paren[\big]{o(n)}$.  Our result establishes the same consequence in the restricted case.
\begin{theorem} 
\label{thm: match-restrbip}For every $c\in\mathbb{N}$, the problem of counting
(edge-colorful or uncolored) $k$-matchings is $\sw$-complete,
even for bipartite graphs of girth at least $c$ whose right side vertices have 
degree at most two.
Furthermore, if $\cc{\#ETH}$ holds, neither of these problems has an algorithm running in time $f(k) \cdot n^{o(k/ \log k )}$, for any computable function $f$.
\end{theorem}
We sketch the proof in~\S\ref{sec:matchings subdivided} by extending the 
so-called Holant problems \cite{Valiant2008,DBLP:journals/jcss/CaiL11}
to an edge-colored variant that proves to be useful for parameterized counting problems. In classical Holant 
problems,
we are given as input a graph $G=(V,E)$ with a signature $f_{v}$ at each vertex $v\in V$. Here,
$f_{v}$ is a function $f_v:\{0,1\}^{I(v)}\to\mathbb{C}$, where~$\{0,1\}^{I(v)}$
is the set of binary assignments to the edges incident with~$v$.
The problem is to compute $\Holant(G)$, a sum over all binary assignments~$x\in\{0,1\}^{E}$, where each assignment $x$ contributes a weight~$\prod_{v\in V}f_{v}(x)$.

In our edge-colored setting, the edges of~$G$ are colored with~$k$ colors and $\Holant(G)$ 
ranges only over assignments of Hamming weight~$k$, picking exactly one edge from each color.
We apply the technique of combined signatures~\cite{DBLP:conf/focs/CurticapeanX15} in this setting, an approach that is 
also implicit in \cite{curticapean2014complexity}.
This way, we will reduce from 
counting edge-colorful $k$-matchings in general graphs
to~$2^{k}$ instances of the restricted bipartite case. Previously,
combined signatures were used only for problems with structural
parameterizations, such as counting perfect matchings in graphs whose
genus or apex number is bounded~\cite{DBLP:conf/focs/CurticapeanX15}. Our 
edge-colorful approach allows us to apply them also when the parameter is the 
solution size~$k$.

\subsubsection{Line graphs}

Building upon Theorem~\ref{thm: match-restrbip}, we prove that
counting $k$-matchings in line graphs is $\sw$-complete and we establish a lower bound under $\cc{\#ETH}$.
\begin{theorem}
\label{thm: match-line}The problem of counting $k$-matchings in line graphs is
$\sw$-complete. Furthermore, if $\cc{\#ETH}$ holds, this problem does not have an $f(k) \cdot n^{o(k/ \log k )}$ time algorithm, for any computable function $f$.
\end{theorem}

Line graphs can be characterized by a finite set of forbidden induced subgraphs \cite{harary2004graph,Sol94}.
They can be recognized in linear time \cite{DBLP:journals/jacm/Lehot74}, and several classical 
$\cc{NP}$-complete problems are polynomial-time solvable in line graphs, such 
as finding a maximum independent set~\cite{zbMATH03693322}, a maximum cut~\cite{DBLP:journals/dam/Guruswami99}, or a maximum clique~\cite{Lozin200574}.
In contrast, Theorem~\ref{thm: match-line} shows that counting $k$-matchings 
remains $\sw$-hard in line graphs.

To prove Theorem~\ref{thm: match-line}, one might try to first prove hardness of counting edge-colorful\linebreak $k$-matchings in line graphs, and then reduce this problem via inclusion--exclusion to the uncolored case. This approach however fails: While the colored problem is easily shown to be 
$\sw$-complete (even on complete graphs), we cannot use inclusion--exclusion to subsequently reduce to counting uncolored matchings, since doing so would lead to graphs that 
are not necessarily line graphs.
Hence we do not know how to prove Theorem~\ref{thm: match-line} via the framework 
of edge-colorful Holant problems introduced before.

Instead, we prove Theorem~\ref{thm: match-line} in \S\ref{sec:P2pack} by means of a delicate interpolation argument 
that is reminiscent of the first hardness proof for uncolored $k$-matchings~\cite{Curticape2013}. 
We generate~a~linear system of equations such that one of the unknowns corresponds to a $\sw$-hard problem.
The right-hand side of the system can be evaluated by means of a gadget construction and an oracle for counting $k$-matchings in line graphs.
It turns out that the system does not have full rank, 
yet a careful analysis shows that the $\sw$-hard unknown we are interested in can still be uniquely determined in polynomial time.

\subsubsection{Perfect matchings in line graphs}
Complementing the above result, we show that the problem of counting \emph{perfect} matchings is $\sp$-hard on line graphs. This holds even for line graphs of bipartite graphs.
Such graphs are known to be perfect and play an important role in the proof of the strong perfect graph theorem~\cite{chudnovsky2006strong}.

\begin{theorem} 
\label{thm:perf_line}
The problem of counting perfect matchings is $\sp$-complete even for graphs that 
have maximum degree~$4$ and are line graphs of bipartite graphs.
On the other hand, the problem is polynomial-time solvable in $3$-regular line graphs.
\end{theorem}

The theorem can be shown by invoking known results for Holant problems~\cite{Cai.Lu.Xia2011},
but in the present paper, we give a self-contained proof.
To this end, we reduce the positive case of Theorem~\ref{thm:perf_line} to a known tractable case of counting constraint satisfaction problems~\cite{DBLP:journals/iandc/CreignouH96},
which admits a simple polynomial-time algorithm.
The negative case of Theorem~\ref{thm:perf_line} 
follows by a relatively straightforward reduction from counting perfect matchings in $3$-regular graphs, 
using specifically tailored gadgets which ensure that the resulting graphs are line graphs.

\subsection{Counting edge-injective homomorphisms}

In our proof of Theorem~\ref{thm: match-line}, we actually prove the equivalent
statement that counting edge-injective homomorphisms from the graph
$k\cdot P_{2}$ to host graphs $G$ is $\sw$-complete.
Here, we write~$k\cdot P_{2}$ for the graph consisting of $k$ disjoint
copies of the path $P_{2}$ with two edges. A~homomorphism~$f$ from $H$ to 
$G$ is \emph{edge-injective} if, for any distinct
(but not necessarily disjoint) edges $e=uv$ and $e'=u'v'$ of $H$,
the edges $f(u)f(v)$ and $f(u')f(v')$ in~$G$ are distinct (but
not necessarily disjoint).
The number of edge-injective homomorphisms from $k\cdot P_{2}$ to~$G$ is equal to the number of $k$-matchings in the line graph~$L(G)$, up to a simple factor depending only on $k$.

Starting from their relevance in our proof of Theorem~\ref{thm: match-line}, 
we observe that edge-injective homomorphisms are an interesting concept on its own, 
since they constitute an intermediate step between homomorphisms and subgraph embeddings,
which are vertex-injective homomorphisms.
To study the complexity of counting edge-injective homomorphisms from general patterns, we define the problems $\#\eiHom(\mathcal{H})$ for fixed graph
classes $\mathcal{H}$: Given graphs $H\in\mathcal{H}$ and $G$,
the problem is to count the edge-injective homomorphisms from $H$ to $G$. 

Similar frameworks exist
for counting subgraphs \cite{DBLP:conf/stoc/CurticapeanDM17,curticapean2014complexity}, 
counting/deciding colorful subgraphs~\cite{curticapean2014complexity,DBLP:journals/dam/Meeks16,Grohe.Schwentick2001},
counting/deciding induced subgraphs \cite{Chen.Thurley2008},
counting/deciding (not necessarily
edge-injective) homomorphisms \cite{Grohe2007,Dalmau.Jonsson2004},
and counting locally-injective homomorphisms~\cite{Roth17}.
In all 
of these cases, precise dichotomies
are known for the parameterized complexity of the problem when the
pattern is chosen from a fixed class $\mathcal{H}$ and the parameter is~$|V(H)|$.
For instance, homomorphisms from $\mathcal{H}$ can be counted in polynomial time 
if~$\mathcal{H}$ has bounded treewidth, and the problem is $\sw$-complete otherwise 
\cite{Dalmau.Jonsson2004}. A similar statement holds for the decision 
version of this problem, but here only the cores of the graphs in $\mathcal{H}$ 
need to have bounded treewidth~\cite{Grohe2007}.

Our main outcome
is a similar result for counting edge-injective homomorphisms.
Let the \emph{weak vertex-cover number} of a graph $G$
be defined as the size of the minimum vertex-cover in the graph obtained
from $G$ by deleting all \emph{isolated edges}, that is, connected components
with two vertices. Furthermore, a graph class $\mathcal{H}$ is \emph{hereditary} if 
$H\in\mathcal{H}$ implies $F\in\mathcal{H}$ for all induced subgraphs
$F$ of~$H$.
\begin{samepage}
\begin{theorem} \label{thm:dichotomy} 
Let $\mathcal{H}$ be any class of graphs. The problem $\#\eiHom(\mathcal{H})$
can be solved in polynomial time if there is a constant $c\in\mathbb{N}$
such that the weak vertex-cover number of all graphs in $\mathcal{H}$
is bounded by $c$. If no such constant exists and $\mathcal{H}$ additionally
is hereditary, then $\#\eiHom(\mathcal{H})$ is $\sw$-complete.
\end{theorem}
\end{samepage}

We prove this theorem in \S\ref{sec:homs}.
For the algorithm, we use ideas from the framework of graph motif parameters~\cite{DBLP:conf/stoc/CurticapeanDM17} to reduce the problem to subgraph counting.
The latter has known~$n^{\vc H+O(1)}$ time algorithms~\cite{williams2013finding,curticapean2014complexity},
where $\vc H$ denotes the vertex-cover number of the pattern graph $H$.
For the hardness result, we use a Ramsey argument to show that any graph class 
with unbounded weak vertex-cover number contains one of six hard classes as 
induced subgraphs. This gives a full dichotomy for the complexity of 
$\#\eiHom(\mathcal{H})$
on hereditary graph classes~$\mathcal{H}$,
but it leaves open the $\sw$-hardness of non-hereditary classes~$\mathcal{H}$.
However, for the particular non-hereditary classes of paths and cycles,
we obtain a separate hardness result.

\begin{theorem} 
The problem $\#\eiHom(\mathcal{H})$ is $\sw$-complete if $\mathcal{H}$
is the class of all paths or the class of all cycles.
\end{theorem}

\section{Preliminaries}\label{sec:prelims}%

A \emph{parameterized counting problem} is a function $\Pi:\zo^*\to\N$ that is 
endowed with a computable \emph{parameterization} $\kappa:\zo^*\to\N$;
it is \emph{fixed-parameter tractable~(FPT)} if there is a computable function 
$f:\N\to\N$ and an $f(k)\cdot\poly(n)$-time algorithm to compute $\Pi(x)$, where~$n=\abs{x}$ and~$k=\kappa(x)$.

An \emph{fpt Turing reduction} is a Turing reduction from a problem 
$(\Pi,\kappa)$ to a problem~$(\Pi',\kappa')$, such that the reduction runs in 
time $f(k)\cdot\poly(n)$ and each query~$y$ to the oracle satisfies~$\kappa'(y)\le g(k)$. Here, both $f$ and $g$ are computable functions.
A problem is \emph{$\sw$-hard} if there is an fpt Turing reduction from the 
problem of counting the cliques of size $k$ in a given graph; since it is believed that 
the latter does not have an FPT-algorithm, $\sw$-hardness is a strong indicator 
that a problem is not FPT. For more details, see~\cite{FlumGrohe}.

The \emph{counting exponential-time hypothesis ({\#ETH})} claims that 
there exists a constant~$\epsilon>0$ for which there is no $\exp\paren{\epsilon n}$ time 
algorithm to compute the number of satisfying assignments for an $n$-variable 
$3$-CNF formula.
An algorithm with running time $f(k)\cdot n^{o(k)}$ to count $k$-cliques in a given graph,
for any function~$f$,
would refute the counting exponential-time hypothesis~\cite{Chen.Chor2005}.

Graphs in this paper are undirected, loop-free, and simple, unless stated otherwise.
Let $H$ and $G$ be graphs.
A function $\varphi: V(H) \rightarrow V(G)$ is a \emph{homomorphism} from~$H$ 
to~$G$ if $\varphi(e)\in E(G)$ holds for all $e\in E(H)$, where 
$\varphi(\set{u,v})=\set{\varphi(u),\varphi(v)}$.
The set of all homomorphisms from~$H$ to~$G$ is denoted by $\Hom(H,G)$.
A homomorphism $\varphi\in\Hom(H,G)$ is called \emph{edge-injective} if all~$e,f \in E(H)$ with $e \neq f$ satisfy~${\varphi(e)\neq\varphi(f)}$.
$\eiHom(H,G)$ denotes the set of all edge-injective homomorphisms from~$H$ to~$G$.
A homomorphism $\varphi\in\Hom(H,G)$ is an \emph{embedding} of $H$ in $G$ if it is injective (on the vertices of~$H$). The set of all embeddings from $H$ to $G$ is denoted by $\Emb(H,G)$.

For a class $\ch$ of graphs, let $\#\eiHom(\classH)$ denote the following computational
problem: Given $H\in\classH$ and a graph $G$, compute the number $\#\eiHom(H,G)$. 
We consider this problem to be parameterized by~$|V(H)|$. 
The problems $\#\Hom(\ch)$ and $\#\Emb(\ch)$
are defined analogously.

The \emph{line graph} $L(G)$ of~$G$ is the graph whose vertex
set satisfies $V(L(G))=E(G)$ such that~$e,f\in E(G)$ with $e\ne f$ are adjacent 
in $L(G)$ if and only if the edges~$e$ and~$f$ are incident to the same vertex in 
$G$. 

\section{Matchings in restricted bipartite graphs}\label{sec:matchings 
  subdivided}%

In this section, we prove Theorem~\ref{thm: match-restrbip}. 
We use \emph{$k$-edge-colored graphs} for $k\in\mathbb{N}$, which are graphs~$G$ with 
a (not necessarily proper) edge-coloring ${c:E(G)\to\set{1,\dots,k}}$. A matching in $G$ is 
\emph{colorful} if it contains exactly one edge from each color.
We let $\pCol(G)$ be the number of such matchings and $\pCol$ be the 
corresponding computational problem; this problem is known to be $\sw$-hard.

\begin{theorem}[\cite{curticapean2014complexity}, Theorem~1.2]
\label{thm: general edge-col-match}%
The problem $\pCol$ is $\sw$-complete.
If $\cc{\#ETH}$ holds, it cannot be solved in time $f(k)\cdot n^{o(k/ \log 
  k)}$ for any computable $f$.
\end{theorem}

A straightforward application of the inclusion--exclusion principle
yields a reduction from counting edge-colorful matchings to counting uncolored matchings 
(see, e.g., \cite[Lemma~1.34]{Curticapean.PhD} or \cite[Lemma~2.7]{curticapean2014complexity}).
\begin{lemma}
  \label{lem: remove-colors}%
  There is an fpt Turing reduction from $\pCol$ for $k$-edge-colored graphs $G$ to 
  the problem of counting $k$-matchings in uncolored subgraphs of $G$; 
  the reduction makes at most~$2^k$ queries, each query is a subgraph of~$G$, 
  and the parameter of each query is~$k$.
\end{lemma}

\subsection{\label{sec:Basic-definitions}Colorful Holant problems}

We first adapt Holant problems to an edge-colorful setting 
by introducing \emph{colorful Holant problems}.
In the uncolored setting, the notion of a ``Holant'' was introduced by
Valiant~\cite{Valiant2008} and later developed to a general theory of Holant problems by Cai, Lu, Xia, and other authors~\cite{DBLP:journals/jcss/CaiL11,Cai.Lu2008}.
In \S\ref{sec: bip-match-proof}, we use colorful Holants to prove 
Theorem~\ref{thm: match-restrbip} by a reduction from $\pCol$.
A more general exposition of this material appears in the first author's PhD thesis 
\cite[Chapters 2 and~5.2]{Curticapean.PhD}.

\global\long\def\assignment{x\in\{0,1\}^{E(\Omega)}}
\begin{definition}
\label{def: signatures, holants}
For a graph~$G$ and a vertex~$v\in V(G)$, 
we denote by~$I(v)$ the set of edges incident to $v$.
For $k\in\mathbb{N}$, a \emph{$k$-edge-colored
signature graph} is a $k$-edge-colored graph~$\Omega$ that has a \emph{signature} 
$f_{v}:\{0,1\}^{I(v)}\to\mathbb{Q}$
at each vertex $v\in V(\Omega)$. The graph underlying $\Omega$ may feature parallel edges.
We write $E_{i}$ for the set of edges with color~$i\in\set{1,\dots,k}$.
\end{definition}

The task of Holant problems is to count,
on input a signature graph $\Omega$,
the Boolean-valued assignments to $E(\Omega)$ that satisfy all local constraints given by the signatures.
In our colorful setting, only colorful assignments will be counted.

\begin{definition}
An assignment $x\in\{0,1\}^{E(\Omega)}$ is \emph{colorful} if, for
each $i\in\set{1,\dots,k}$, there is exactly one edge $e\in E_{i}$ with~$x(e)=1$.
Given a set $S\subseteq E(\Omega)$, we write $x|_{S}$ for the
restriction of $x$ to $S$, which is the unique assignment in $\{0,1\}^{S}$
that agrees with $x$ on $S$.
We define $\ColHolant(\Omega)$ as the sum 
\[
\ColHolant(\Omega)=
\sum_{\substack{x\in\{0,1\}^{E(\Omega)}\\
\text{colorful}
}\;
}\prod_{v\in V(\Omega)}f_{v}\paren*{x|_{I(v)}}
\,.
\]
\end{definition}

Next, we express the number of edge-colorful matchings in a graph as a colorful Holant problem.
If all signatures in $\Omega$ map to $\{0,1\}$, then $\ColHolant(\Omega)$
simply counts the edge-colorful assignments $x$ that satisfy $f_{v}(x|_{I(v)})=1$
for all $v\in V(\Omega)$.
For assignments $x\in\{0,1\}^{*}$, write $\hw(x)$ for the Hamming
weight of $x$. For a proposition~$\varphi$, let~$[\varphi]$ be defined to be $1$ if $\varphi$ holds and $0$ otherwise.

\begin{fact}\label{fact: Match-Holant}%
Let $k\in\mathbb{N}$ and let $G$ be a $k$-edge-colored graph. Define the 
$k$-edge-colored signature graph $\Omega=\Omega(G)$
by associating with each vertex $v\in V(G)$ the signature
\begin{align*}
\sigHW{\leq1}:~&\{0,1\}^{I(v)}\to\{0,1\}\\
~&x \mapsto [\hw(x)\leq 1]\,.
\end{align*}
Then we have \[\ColHolant(\Omega)=\pCol(G).\]

This fact can also be used in reverse:
For any signature graph $\Omega$ that has $\sigHW{\leq1}$ associated with every vertex,
we can obtain a graph $G$ such that $\ColHolant(\Omega)=\pCol(G)$ by deleting the signatures from $\Omega$.
\end{fact}

If a signature graph $\Omega$ has a vertex
$v$ with some complicated signature $f$ associated with it, we can sometimes simulate the effect of 
$f$ by replacing~$v$ with a graph fragment that has 
only simpler signatures associated with its vertices,
such as $\sigHW{\leq1}$.
Replacing all signatures by such graph fragments,
we obtain a signature graph $\Omega'$ featuring only $\sigHW{\leq1}$.
Then, by Fact~\ref{fact: Match-Holant},
the quantity $\ColHolant(\Omega')$ can be expressed as a number of edge-colorful matchings.
This will allow us to reduce the computation of $\ColHolant(\Omega)$ to an instance of $\pCol$.
The graph fragments required for this reduction are formally defined as \emph{edge-colored matchgates}:
\begin{definition}\label{def: module-contract}%
An \emph{edge-colored matchgate} is an edge-colored signature graph
$\Gamma$ that contains a set $D\subseteq E(\Gamma)$ of \emph{dangling}
edges. These are edges with only one endpoint in $V(\Gamma)$, and we consider them to be labeled with $1,\ldots,|D|$.
Furthermore, we require the signature $\sigHW{\leq1}$ to be associated with all vertices in $\Gamma$.
The colors on edges in $E(\Gamma)\setminus D$ are called \emph{internal} colors.

We say that an assignment $y\in\{0,1\}^{E(\Gamma)}$ extends an assignment~$x\in\{0,1\}^{D}$
if $y$ agrees with~$x$ on~$D$. The signature $\ColSig(\Gamma):\{0,1\}^{D}\to\mathbb{Q}$
of $\Gamma$ is defined via
\[
\ColSig(\Gamma,x)=\sum_{\substack{y\in\{0,1\}^{E(\Gamma)}\\
\text{colorful, extends $x$}
}
}\prod_{v\in V(\Gamma)}f_{v}\paren*{y|_{I(v)}}\,.
\]

Let $\Omega$ be a $k$-edge-colored signature graph
and let $\Gamma$ be an edge-colored matchgate with $d \in \mathbb N$ dangling edges
such that the internal colors of $\Gamma$ are disjoint from 
the colors $\set{1,\dots,k}$ present in $\Omega$.
Then we can \emph{insert} $\Gamma$ at a vertex~$v \in V(\Omega)$ of degree $d$ as follows,
see Figure~\ref{fig:matchgate}:
First delete~$v$ from~$\Omega$, but keep~$I(v)$ as dangling edges in~$\Omega$,
and choose an ordering $e_1, \ldots , e_d$ of $I(v)$.  
Then insert a disjoint copy of~$\Gamma$ into~$\Omega$,
and for all $i \in [d]$, 
identify the $i$-th dangling edge of $\Gamma$ with~$e_i$.
That is, if $e_i$ has endpoint $u_i$ in $\Omega$
and the $i$-th dangling edge of $\Gamma$ has endpoint $v_i$ in~$\Gamma$,
then form the edge $u_i v_i$ in the resulting graph.%
\footnote{We need to assume here that the colors of the two identified edges agree.}%

\end{definition}
\begin{figure}
  \centering
\begin{tikzpicture}[-,thick,scale=0.6,main node/.style={circle,inner sep=1.5pt,fill}]

   \node[main node] (1) at (0,3) {};
   \node[main node] (2) at (3.5,3) {};
   \node[main node] (3) at (-1,0) {};
   \node[main node] (4) at (4,-0.5) {};
   \node[main node, label=below:{$v$}] (5) at (2,1) {};
   \node[main node] (6) at (8,2) {};
   \node[main node] (7) at (7.5,1) {};
   \node[main node] (8) at (8.5,1) {};
   \node[main node] (9) at (12.5,3) {};
   \node[main node] (10) at (16,3) {};
   \node[main node] (11) at (11.5,0) {};
   \node[main node] (12) at (16.5,-0.5) {};
   \node[main node] (13) at (14.5,2) {};
   \node[main node] (14) at (14,1) {};
   \node[main node] (15) at (15,1) {};
   
   \node (42) at (0,-1) {\large $\Omega$};
   \node (42) at (8,-1) {\large $\Gamma$};

   \draw (1) -- (2);
   \path[dotted] (2) edge node [below,right] {$e_1$} (5);
   \path[dashed, bend left] (1) edge node [right] {$e_2$} (5);
   \path[dashed, bend right] (1) edge node [left] {$e_3$} (5);
   \draw (1) -- (3);
   \draw (2) -- (4);
   \draw (4) -- (3);
   \path[dashdotted] (3) edge node [above] {$e_4$} (5);
   
   \path[dotted] (6) edge node [above] {\large $1$} (9.8,2.25);
   \path[dashed] (6) edge node [above right] {\large $2$} (7,3);
   \path[dashed] (7) edge node [above right] {\large $3$} (6,2);
   \path[dashdotted] (7) edge node [below] {\large $4$} (6,-0.25);
   \draw (6) -- (8);
   \draw (6) -- (7);
   \draw (7) -- (8);
   
   \draw (9) -- (10);
   \draw (9) -- (11);
   \draw (12) -- (10);
   \draw (12) -- (11);
   \draw (13) -- (15);
   \draw (13) -- (14);
   \draw (14) -- (15);
   \draw[dashed] (13) -- (9);
   \draw[dotted] (13) -- (10);
   \draw[dashed] (14) -- (9);
   \draw[dashdotted] (11) -- (14);
\end{tikzpicture}
\par

\caption{\label{fig:matchgate}A matchgate $\Gamma$ is inserted into a signature graph $\Omega$ at vertex $v$.}
\end{figure}
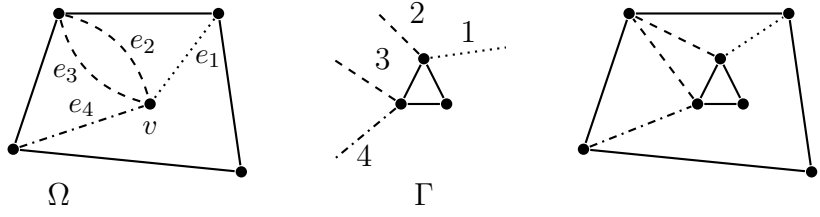

%%%%%%%%%%%%%%%%%%%%%%%%%%%%%%%%%%%%%%%%%%%%%%%%%Worked through it until here%%%%%%%%%%%%
\begin{remark}
When inserting $\Gamma$ into a signature graph $\Omega$, we implicitly assume that the edge-colors of dangling edges are a subset of the edge-colors in $\Omega$.
\end{remark}

A simple calculation shows that 
inserting a matchgate $\Gamma$ at a vertex $v$ with $f_v = \ColSig(\Gamma)$
in a signature graph $\Omega$ preserves the value of $\ColHolant(\Omega)$.
By repeating this operation, we obtain the following fact, as
proved in Fact~2.17 and Lemma~5.16 of~\cite{Curticapean.PhD}.
\begin{fact}
\label{fact: realize signature graph with matchgates}
Let $\Omega$ be a $k$-edge-colored signature graph,
and for each $v\in V(\Omega)$, 
let $f_{v}$ be the signature associated with $v$. If there is a matchgate $\Gamma_{v}$
with ${\ColSig(\Gamma_{v})=f_{v}}$ for every vertex~$v$, then we can efficiently 
construct an edge-colored graph $G$ on $O(\sum_{v}|V(\Gamma_{v})| + \sum_{v}|E(\Gamma_{v})|)$ vertices
and edges such that $\ColHolant(\Omega)=\pCol(G)$.
\end{fact}

If the involved signatures cannot be realized by matchgates, then Fact~\ref{fact: realize signature graph with matchgates} is not applicable.
For such cases, Curticapean and Xia~\cite{DBLP:conf/focs/CurticapeanX15} define \emph{combined signatures}:
Even if a given signature~$f$ cannot be realized via matchgates, we may be able to express $f$ as a linear combination of $t\in\mathbb{N}$ signatures that do admit matchgates.
If there are $s\in\mathbb{N}$ occurrences
of such signatures in $\Omega$, then we can compute $\ColHolant(\Omega)$
as a linear combination of $t^{s}$ colorful Holants, where all involved
signatures can be realized by matchgates. 
In the following, we write $[N]=\set{1,\dots,N}$.
\begin{lemma} 
\label{lem: combined signature lemma}Let $\Omega$ be a $k$-colored
signature graph. Let $s,t\in\mathbb{N}$ and let $w_{1},\ldots,w_{s}$
be fixed distinct vertices of $\Omega$ such that the following holds: For
all $\kappa\in[s]$, the signature~$f_{\kappa}$ at $w_{\kappa}$
admits coefficients $c_{\kappa,1},\ldots,c_{\kappa,t}\in\mathbb{Q}$
and signatures $g_{\kappa,1},\ldots,g_{\kappa,t}$ such that 
$f_{\kappa}=\sum_{i=1}^{t}c_{\kappa,i}\cdot g_{\kappa,i}$ holds point-wise. 

Given a tuple $\theta\in[t]^{s}$, let $\Omega_{\theta}$ be the edge-colored signature graph defined
by replacing, for\linebreak each $\kappa\in[s]$, the signature $f_{\kappa}$
at $w_{\kappa}$ with $g_{\kappa,\theta(\kappa)}$. Then we have 
\begin{equation}
\ColHolant(\Omega)=\sum_{\theta\in[t]^{s}}\left(\prod_{\kappa=1}^{s}c_{\kappa,\theta(\kappa)}\right)\cdot\ColHolant(\Omega_{\theta})\,.\label{eq: combined sig lin-comb}
\end{equation}
\end{lemma}

\begin{proof}
In the following claim, we first prove the lemma for $s=1$. That is, the signature of exactly one 
vertex~$w_1$ is expressed as a linear combination of signatures.
\begin{claim}
\label{lem:Holant-distrib-once}Let $\Omega$ be a signature graph and let $w\in V(\Omega)$ be a fixed vertex with signature
$f_{w}$. Let $g_{1},\ldots,g_{t}$ be signatures and $c_{1},\ldots,c_{t}\in\mathbb{Q}$ be such that 
$f_{w}=\sum_{i=1}^{t}c_{i}\cdot g_{i}$ holds point-wise.
For $i\in[t]$, let $\Omega_{i}$ be the signature graph obtained
from $\Omega$ by replacing $f_{w}$ with $g_{i}$. Then 
$\ColHolant(\Omega)=\sum_{i=1}^{t}c_{i}\cdot\ColHolant(\Omega_{i}).$
\end{claim}
\begin{claimproof}
\begin{samepage}
In the following, let $x$ range over edge-colorful assignments in $\{0,1\}^{E(\Omega)}$.
By elementary manipulations, we have
\begin{alignat*}{2}
\ColHolant(\Omega) & = &  & \sum_{x}f_{w}(x)\prod_{v\in V(\Omega)\setminus\{w\}}f_{v}(x)\\
 & = &  & \sum_{x}\left(\sum_{i=1}^{t}c_{i}\cdot g_{i}(x)\right)\prod_{v\in V(\Omega)\setminus\{w\}}f_{v}(x)\\
 & = &  & \sum_{i=1}^{t}\sum_{x}c_{i}\cdot g_{i}(x)\prod_{v\in V(\Omega)\setminus\{w\}}f_{v}(x)\\
 & = &  & \sum_{i=1}^{t}c_{i}\cdot\ColHolant(\Omega_{i}).
\end{alignat*}
\end{samepage}

\end{claimproof}
The lemma follows by applying Claim~\ref{lem:Holant-distrib-once} inductively for $w_{1},\ldots,w_{s}$.
Each of the $s$ involved steps reduces the number of combined signatures by one, and elementary
algebraic manipulations imply (\ref{eq: combined sig lin-comb}).
\end{proof}

Lemma~\ref{lem: combined signature lemma} allows us to prove hardness results 
under fpt Turing reductions if $\ColHolant(\Omega)$ is $\sw$-hard to compute and 
the values $\ColHolant(\Omega_{\theta})$ for all $\theta$ can
be computed by reductions to the target problem.
This is our approach in the remainder of this section.

\subsection{\texorpdfstring{$k$-Matchings}{k-Matchings} in bipartite graphs}
\label{sec: bip-match-proof}

We prove Theorem~\ref{thm: match-restrbip} by a reduction from $\pCol$, 
which is $\sw$-complete by Theorem~\ref{thm: general edge-col-match}. 
Let ${k\in\mathbb{N}}$ and let~$G$ be a
simple $k$-edge-colored graph for which we want to compute the\linebreak number $\pCol(G)$.
To this end, we first construct a bipartite signature graph~$\Omega_{\mathit{bip}}$
such that~$\ColHolant(\Omega_{\mathit{bip}})=\pCol(G)$ holds.
\begin{lemma} 
\label{lem: Omega-bip}Given a $k$-edge-colored graph $G$, let
$\Omega_{\mathit{bip}}=\Omega_{\mathit{bip}}(G)$ denote the signature
graph with edge-colors $[k]\times[2]$ that is constructed as follows:
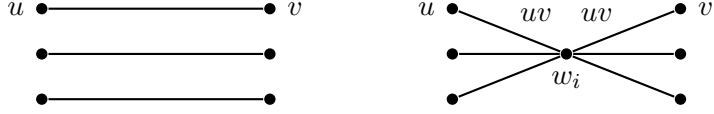
\begin{figure}
  \centering
\begin{tikzpicture}[-,thick,scale=0.6,main node/.style={circle,inner sep=1.5pt,fill}]

   \node[main node,label=left:{$u$}] (1) at (-2,3) {};
   \node[main node,label=right:{$v$}] (2) at (3,3) {};
   \draw (1) -- (2);
   \node[main node] (3) at (-2,2) {};
   \node[main node] (4) at (3,2) {};
   \draw (3) -- (4);
   \node[main node] (5) at (-2,1) {};
   \node[main node] (6) at (3,1) {};
   \draw (5) -- (6);
   
   \node[main node,label=below:{$w_i$}] (13) at (9.5,2) {};
   \node[main node,label=left:{$u$}] (7) at (7,3) {};
   \node[main node,label=right:{$v$}] (8) at (12,3) {};
   \path (13) edge node [above right] {$uv$} (7);
   \path (13) edge node [above left] {$uv$} (8);
   \node[main node] (9) at (7,2) {};
   \node[main node] (10) at (12,2) {};
   \draw (13) -- (9);
   \draw (13) -- (10);
   \node[main node] (11) at (7,1) {};
   \node[main node] (12) at (12,1) {};
   \draw (13) -- (11);
   \draw (13) -- (12);
\end{tikzpicture}
\par

\caption{\label{fig:bip-graph}Edges of color $i$ are shown on the left, and
their corresponding edges in the graph $\Omega_{\mathit{bip}}$ together with their annotations are
shown on the right.}
\end{figure}
Initially, $\Omega_{\mathit{bip}}$ is $G$, where each vertex is associated 
    with the signature $\sigHW{\leq1}$. Then, for each $i\in[k]$:

  \begin{enumerate}
    \item Add a fresh vertex $w_{i}$ to $\Omega_{\mathit{bip}}$.
    \item For each $e\in E(G)$ of color~$i$ and with $e=uv$, delete $e$
      and insert an edge $uw_{i}$ of color~$(i,1)$ and an edge $w_{i}v$
      of color $(i,2)$. Annotate the added edges with $\pi(uw_{i})=\pi(w_{i}v)=e$.
    \item Note that every colorful assignment $x\in\{0,1\}^{I(w_{i})}$ at a vertex $w_i$ has
      precisely two edges~$e_1(x)$ and $e_2(x)$ that are incident to $w_i$ and 
      assigned~$1$ by~$x$.  We associate $w_{i}$ with the signature $f_{i}$ that 
      maps
      $x\in\{0,1\}^{I(w_{i})}$ to
      $f_{i}(x) = [\pi(e_{1}(x))=\pi(e_{2}(x))]$.
\end{enumerate}
The constructed signature graph $\Omega_{\mathit{bip}}$ satisfies $\ColHolant(\Omega_{\mathit{bip}})=\pCol(G).$
\end{lemma}

\begin{proof}
It is clear that $\Omega_{\mathit{bip}}$ is bipartite, since every
edge of $\Omega_{\mathit{bip}}$ has exactly one of the vertices~$w_{i}$ for $i\in[k]$ as an endpoint.

Let us call an assignment $x\in\{0,1\}^{E(\Omega_{\mathit{bip}})}$ satisfying for $\Omega_{\mathit{bip}}$ if none of the signatures in~$\Omega_{\mathit{bip}}$ vanishes on $x$.
The edge-colorful satisfying assignments~$x$
correspond bijectively to the edge-colorful matchings of $G$: In any
such~$x$, the vertex~$w_{i}$ for $i\in[k]$ is incident with two
edges~$e_{i}$ and~$e'_{i}$ that are assigned~$1$ under $x$ and have the same annotation
$h_{i}=\pi(e_{i})=\pi(e_{i}')$, for some~${h_{i}\in E(G)}$. We can
hence contract $e_{i}$ and $e'_{i}$ to one edge $h_{i}$. The resulting
edge set is an edge-colorful matching in $G$ due to the signature
$\sigHW{\leq1}$ at non-subdivision vertices. Reversing this contraction operation, every edge-colorful
matching in $G$ can be extended to a unique satisfying assignment
$x\in\{0,1\}^{E(\Omega_{\mathit{bip}})}$, and all signatures evaluate to $1$ on this assignment.
\end{proof}

We now realize the signatures $f_i$ in $\Omega_{\mathit{bip}}$
by linear combinations of the signatures of edge-colored matchgates.
For $i\in [k]$, let $E_i(G)$ denote the $i$-colored edges in $G$.
Let $m_i=|E_i(G)|$ and order the edges in $E_i(G)$ in some arbitrary fixed way.

\begin{definition}
  Recall the definition of $\Omega_{\mathit{bip}}$ from Lemma~\ref{lem: Omega-bip}. For $i\in[k]$, let ${m=m_i}$ and let~$\Gamma_{i,1}$ denote the matchgate on dangling
edges $I(w_{i})$ that consists of $2m$ vertices and is defined as
follows:
\begin{enumerate}
\item Create independent sets $a_{1},\ldots,a_{m}$ and $b_{1},\ldots,b_{m}$,
which we call ``external'' vertices.
\item Then, for all $j\in[m]$ and all edges $e,e'\in E(\Omega_{\mathit{bip}})$
of colors~$(i,1)$ and $(i,2)$ with $\pi(e)=\pi(e')$: If $\pi(e)$
is the $j$-th edge in the ordering of~$E_i(G)$, for $j\in\mathbb{N}$,
then attach $e$ as dangling edge to $a_{j}$ and $e'$ as dangling
edge to $b_{j}$.
\end{enumerate}
Let $\Gamma_{i,2}$ be defined likewise, with the following addition:
For all $j\in[m]$, add an extra vertex~$c_{j}$, an edge $a_{j}c_{j}$
of color $(i,3)$ and an edge $c_{j}b_{j}$ of color $(i,4)$. 
\end{definition}

\begin{lemma} 
\label{lem: comb-sig-omegabip}
We can express the signature $f_i$ from Lemma~\ref{lem: Omega-bip} as the linear combination
\[f_{i}=(m^{2}-3m+3)\cdot\ColSig(\Gamma_{i,1})-\ColSig(\Gamma_{i,2}).\]
\end{lemma}

\begin{proof}
Observe first that, for all edge-colorful
$x\in\{0,1\}^{I(w_{i})}$, we have 
\[
\ColSig_{\mathrm{col}}(\Gamma_{i,1},x)=1.
\]
This is because $x$ trivially is the only satisfying assignment that
extends $x$, since there are no edges other than $I(w_{i})$ in $\Gamma_{i,1}$.

Concerning $\Gamma_{i,2}$, let $x\in\{0,1\}^{I(w_{i})}$ be a colorful
assignment and let $e_{1},e_{2}$ be the edges that are assigned $1$ under $x$. We show 
\begin{equation}
\ColSig_{\mathrm{col}}(\Gamma_{i,2},x)=\begin{cases}
m^{2}-3m+2 & \mbox{if }\pi(e_{1})=\pi(e_{2}),\\
m^{2}-3m+3 & \mbox{otherwise},
\end{cases}\label{eq: sigGamma2}
\end{equation}
which implies the claim of the lemma. In the following, we calculate the two cases in (\ref{eq: sigGamma2})
separately. To this end, let us say that a path in $\Gamma_{i,2}$ is \emph{hit} by $x$ if at least one of the edges assigned $1$ in $x$ is incident with an endpoint of the path.
\begin{itemize}
\item If $\pi(e_{1})=\pi(e_{2})$, then there are $m-1$ paths in $\Gamma_{i,2}$ not hit by
$x$, each on the same two colors. This gives $(m-1)_{2}=m^{2}-3m+2$
matchings.
\item Otherwise, there are $m-2$ paths not hit by $x$. Each matching
may contain 
\begin{itemize}
\item $2$ edges from the intact paths, yielding $(m-2)_{2}$ matchings,
or 
\item $1$ such edge, yielding $2(m-2)$ matchings, or
\item $0$ such edges, yielding $1$ matching.
\end{itemize}
\end{itemize}
By summing over the disjoint possible choices, we obtain (\ref{eq: sigGamma2}).
\end{proof}

Using Lemmas~\ref{lem: combined signature lemma}, \ref{lem: Omega-bip} and \ref{lem: comb-sig-omegabip}, 
we can now reduce counting edge-colorful matchings in graphs~$G$ to the same problem in subdivisions of $G$.
For a $k$-colored graph~$G$ and $t\in\mathbb{N}$,
a $t$-subdivision of $G$ is obtained by replacing each edge of $G$ by a path with exactly~$t$ inner 
vertices. We may assign any colors to the new edges.
\begin{lemma} 
\label{lem: reduce-subdiv}Let $G$ be a $k$-edge-colored graph
on $n$ vertices and $m$ edges. Then we can compute $\pCol(G)$ with
$O(2^{k})$ oracle calls of the form $\pCol(G')$ for graphs $G'$ that are subgraphs
of a $3$-subdivision of $G$. Furthermore, $G'$ has at most $4(n+m)$ vertices
and edges and at most $4k$ colors.\end{lemma}

\begin{proof}
Recall that $\Omega_{\mathit{bip}} = \Omega_{\mathit{bip}}(G)$ is bipartite, with the vertices
$\{w_{i}\}_{i\in[k]}$ on one side.  
If we insert $\Gamma_{i,2}$ at $w_i$, for every $i\in[k]$, the resulting signature graph $\Omega^*$ is a $3$-subdivision of $G$. It has at most $4(n+m)$ vertices and edges and at most $4k$ colors.
For $\theta \in [2]^k$, consider the graph $\Omega_\theta$ obtained from $\Omega_{\mathit{bip}}$ by inserting $\Gamma_{i,\theta(i)}$ at $w_i$ for every $i\in [k]$:
We observe that $\Omega_\theta$ is a subgraph of $\Omega^*$.

Invoking Lemmas~\ref{lem: combined signature lemma} and \ref{lem: comb-sig-omegabip}, we can hence
write $\ColHolant(\Omega_{\mathit{bip}})$ as a linear combination
of the $2^{k}$ quantities $\ColHolant(\Omega_\theta)$ where each $\Omega_\theta$
features only the signature $\sigHW{\leq1}$. By Fact~\ref{fact: Match-Holant}, the value $\ColHolant(\Omega_\theta)$
can hence be computed as $\pCol(G')$ for the edge-colored graph~$G'$
obtained by removing all signatures from $\Omega_\theta$. As observed above, this is a subgraph of a $3$-subdivision of $G$.
\end{proof}

Theorem~\ref{thm: match-restrbip} now follows easily from the hardness of $\pCol$ and repeated
applications of Lemma~\ref{lem: reduce-subdiv}:

\begin{proof}[Proof of Theorem~\ref{thm: match-restrbip}]
Let $c\in \mathbb{N}$ be an arbitrary constant and let $G$ be a $k$-edge-colored graph for which we wish to compute $\pCol(G)$.
We first prove Theorem~\ref{thm: match-restrbip} for the edge-colorful case. 
To this end, we reduce the computation of $\pCol(G)$ to instances $\pCol(F)$, where $F$ is obtained from $G$ by repeating $c$ times the operation of taking a subgraph of the $3$-subdivision.

Note that this indeed proves the edge-colorful part of Theorem~\ref{thm: match-restrbip}:
If $F$ is obtained as above, then $F$ is bipartite, with the maximum degree of one side bounded by $2$: Put the first and third vertex (if present) of all subdivided edges on one side. Concerning the girth, each cycle $C$ in $G$ will appear in $F$ either (i) as a subdivision of $C$, or (ii) not at all, because vertices of $C$ or its subdivisions were deleted in the process. No other cycles can be created by the operation.

A single application of Lemma~\ref{lem: reduce-subdiv}
on a $[t]$-edge-colored graph $F$ creates $2^t$ new instances for $\pCol$ that are
$[4t]$-edge-colored subgraphs of $3$-subdivisions of~$F$. Iteratively applying
Lemma~\ref{lem: reduce-subdiv} on $G$ for a total of $c$ times yields~$2^{O(4^c k)}$ instances $\pCol(F)$ where each $F$ is a $[4^c k]$-edge-colored graph on ${O(4^c(n+m))}$ vertices and edges that is obtained by repeatedly taking subgraphs of subdivisions. In particular, since $c$ is constant, we reduce counting $k$-matchings in $G$ to counting $O(k)$-matchings in the graphs $F$, hence proving the claimed lower bound under~$\cc{\#ETH}$.

Finally, we show $\sw$-completeness of the uncolored variant of the problem. To this end, we use Lemma~\ref{lem: remove-colors} to reduce from the edge-colorful variant: To compute $\pCol(F)$ for a $[t]$-colored graph $G$, we only need to count $t$-matchings in uncolored subgraphs of $F$. This preserves the properties required on $F$ and the claimed lower bound under $\cc{\#ETH}$.
\end{proof}

\section{Matchings in line graphs}\label{sec:P2pack}%

We now sketch the proof of Theorem~\ref{thm: match-line}, 
which asserts that counting $k$-matchings in line graphs is $\sw$-hard.
In our proof, we will use an equivalent characterization of this problem:
A~\emph{wedge} is any graph isomorphic to $P_2$, the path with two edges, and a 
\emph{wedge packing} $k\cdot P_2$ is the vertex-disjoint union of~$k$ wedges.
For any graph~$G$, we observe that the number of embeddings of a $k$-matching 
in~$L(G)$ is equal to the number of edge-injective homomorphisms from a wedge 
packing~$k\cdot P_2$ to~$G$.

To prove Theorem~\ref{thm: match-line}, we reduce from the $k$-matching problem in well-structured bipartite graphs to 
counting edge-injective homomorphisms from wedge packings.
The following lemma encapsulates the interpolation argument 
used in the reduction. 
For~$t\in \mathbb{N}$, let
$(x)_t$
denote the falling factorial, where
\[(x)_t = (x)\cdot(x-1) \dotsm (x-t+1)\,.\]
\begin{lemma} 
  \label{lem:interpolation}
  For all $g,b\in\N$, let $a_{g,b}\in\Q$ be unknowns, and for all $r\in\N$, let
  $P_r(y)$ be the univariate polynomial such that
  \begin{align*}
    P_r(y) = \sum_{k=0}^{r}\sum_{t=0}^{k} a_{t,k-t}\cdot \binom{r}{k} \cdot 
    (y-t)_{2(r-k)}
    \,.
  \end{align*}
  There is a polynomial-time algorithm that, given a number $k$ and the 
  coefficients of~$P_r(y)$ for all $r\in\N$ with $r\le O(k)$, computes the 
  numbers $a_{t,k-t}$ for all $t\in\set{0,\dots,k}$.
\end{lemma}

\begin{proof}
Let $t\in\set{0,\dots,k}$.
  For all $k,i\in\N$, let $I_{k,i}$ be defined as
  \[
    I_{k,i}=\sum_{t=0}^{k} a_{t,k-t} \cdot t^i
    \,.
  \]
  As an intermediate step, we construct a polynomial-time algorithm that allows us, 
  given the coefficients of $P_r(y)$ and a number $m\in\N$, to compute $I_{k,i}$ 
  for all $k,i\in\N$ with $2k+i\le m$.

  If $m=0$, then $I_{0,0}$ is the only number we need to compute. We obtain it by observing that~$P_0(y)=I_{0,0}=a_{0,0}$.
  Now suppose that $m>0$ and that we inductively already computed the values 
  $I_{k,i}$ for all $k,i\in\N$ with $2k+i\le m$.
  We will compute the values~$I_{k,i}$ with $2k+i=m+1$.

  Let $r$ be an integer that satisfies $2r-(m+1) \geq 0$.
  Furthermore, let $C^r$ denote the coefficient of $y^{2r-(m+1)}$ in $P_r(y)$, 
  which is given as input.
  We want to describe~$C^r$ as an expression in terms of the unknowns $a_{g,b}$.
  To this end, we investigate which of the summands 
  $a_{t,k-t}\cdot\binom{r}{k}\cdot (y-t)_{2(r-k)}$ contribute to $C^r$.
  \begin{claim}
		If $k > \lfloor\frac{m+1}{2}\rfloor$ then $2(r-k) < 2r-(m+1)$.
	\end{claim}
	\begin{claimproof}
    We have
    $2(r-k) = 2r-2k < 2r - 2\lfloor\frac{m+1}{2}\rfloor$.
    If $m+1$ is even, we get $2(r-k)<2r - (m+1)$ as claimed.
    Otherwise $m+1$ is odd, and we only get $2(r-k)<2r-m$.
    However, since~$m$ and~$2(r-k)$ are both even, we actually get 
    $2(r-k)<2r-(m+1)$ as claimed.
  \end{claimproof}
	It follows that the summands with $k > \lfloor\frac{m+1}{2}\rfloor$ do not contribute to $C^r$.
 
  Let us view $(y-t)_{2(r-k)}$ as a bivariate polynomial in $y$ and~$t$ for a 
  moment.
  Then, by expanding this polynomial in powers of $y$, there exist univariate polynomials $\sigma_i(t)$ for all $i\in\N$ with $i\le 
  2(r-k)$ such that
	\[
    (y-t)_{2(r-k)} = \sum_{i=0}^{2(r-k)}\sigma_i(t) \cdot y^{2(r-k)-i}
    \,.
  \]
  Using bivariate interpolation, we can easily compute all coefficients of 
  $\sigma_{i}(t)$ for all $i\le 2(r-k)$.
  Note that the coefficient of $t^{m+1-2k}$ in $\sigma_{m+1-2k}$ is 
  $(-1)^{m+1}\cdot\binom{2(r-k)}{m+1-2k}$.
  Let $c_0,\ldots,c_{m+1-2k-1}$ be the remaining coefficients.\newline
  Since only terms with $k\le \lfloor (m+1)/2 \rfloor$ contribute to $C^r$, we 
  obtain the following.
  {\small
  \begin{align*}
    C^r &= \sum_{k=0}^{\lfloor\frac{m+1}{2}\rfloor} \sum_{t=0}^{k}a_{t,k-t} \cdot 
    \binom{r}{k} \cdot \sigma_{m+1-2k}(t)
    \\
    & = \sum_{k=0}^{\lfloor\frac{m+1}{2}\rfloor} \sum_{t=0}^{k}a_{t,k-t} \cdot 
    \binom{r}{k} \cdot \left[(-1)^{m+1}\binom{2(r-k)}{m+1-2k}t^{m+1-2k} + 
    \sum_{j=0}^{m+1-2k-1}c_jt^j \right]\\
    & = (-1)^{m+1}\cdot\sum_{k=0}^{\lfloor\frac{m+1}{2}\rfloor} 
    \Bigg[\sum_{t=0}^{k}a_{t,k-t} \cdot \binom{r}{k} \cdot\binom{2(r-k)}{m+1-2k}\cdot 
    t^{m+1-2k} + \sum_{t=0}^{k}a_{t,k-t} \cdot \binom{r}{k} \cdot 
    \sum_{j=0}^{m-2k}c_j\cdot t^j \Bigg]\\
    & = (-1)^{m+1}\cdot \left[ \sum_{k=0}^{\lfloor\frac{m+1}{2}\rfloor} 
    \binom{2(r-k)}{m+1-2k} \binom{r}{k} \cdot \sum_{t=0}^{k} a_{t,k-t} 
    \cdot t^{m+1-2k} \right]
     + \left[ \sum_{k=0}^{\lfloor\frac{m+1}{2}\rfloor}\sum_{j=0}^{m-2k}c_j 
     \binom{r}{k} \sum_{t=0}^{k}a_{t,k-t}\cdot t^j \right]\\
    & = (-1)^{m+1}\cdot \left[ \sum_{k=0}^{\lfloor\frac{m+1}{2}\rfloor} 
    \binom{2(r-k)}{m+1-2k} \cdot \binom{r}{k} \cdot I_{k,m+1-2k} \right]
     + \left[ \sum_{k=0}^{\lfloor\frac{m+1}{2}\rfloor}\sum_{j=0}^{m-2k}c_j \cdot 
    \binom{r}{k} \cdot I_{k,j} \right] \,.
  \end{align*}}
  Now consider the $I_{k,j}$ from the last sum.
  Since $2k+j \leq 2k + m -2k = m$, we have already computed these $I_{k,j}$ 
  recursively.
  We also know all of the $c_j$, so we can compute the number $C'^{r}$ for any 
  $r \geq \frac{m+1}{2}$, where
	\[
    C'^r =  \sum_{k=0}^{\lfloor\frac{m+1}{2}\rfloor} \binom{2(r-k)}{m+1-2k} 
    \cdot \binom{r}{k} \cdot I_{k,m+1-2k}\,.
  \]
	Finally, consider the corresponding matrix $A$ such that
  \[A_{j,i}= \binom{2(r_j-i)}{m+1 - 2i}\cdot \binom{r_j}{i}\]
  for $i = 0,\ldots,\lfloor\frac{m+1}{2}\rfloor$ and pairwise distinct and large 
  enough $r_0,\ldots,r_{\lfloor\frac{m+1}{2}\rfloor}$.\newline  Column~$i$ is an  
  evaluation vector of the polynomial $Q_i(r) = \binom{2(r-i)}{m+1 - 2i} \cdot 
  \binom{r}{i}$. Each $Q_i$ has degree $m+1-2i + i = m+1-i$; in particular, the degree of $Q_i$ 
  is different for different $i$.
  This implies that the set $\setc{Q_i}{i\in\N}$ is a set of linearly independent 
  polynomials, and thus the column vectors of $A$ are linearly independent and 
  $A$ is invertible.
  This allows us to compute the unique solution for the $I_{k,m+1-2k}$ for all 
  $k\in\N$ with $k\le m/2$.

  Finally, we argue how to compute the $a_{t,k-t}$ from the $I_{k,i}$.
  By definition, we have the following set of linear equations:
	\begin{equation*}
	\begin{split}
	\sum_{t=0}^{k}t^0 \cdot a_{t,k-t} & = I_{k,0}\\
	\sum_{t=0}^{k}t^1 \cdot a_{t,k-t} & = I_{k,1}\\
	~& \cdots  \\
	\sum_{t=0}^{k}t^k \cdot a_{t,k-t} & = I_{k,k}
	\end{split}
	\end{equation*}
  The corresponding matrix $B$ where $(B)_{i,j} = j^i$ for $i,j=0,\ldots,k$ is a 
  Vandermonde matrix and thus invertible.
  Therefore we can compute the unique solution for the~$a_{t,k-t}$.
\end{proof}

We then prove Theorem~\ref{thm: match-line} by showing the following equivalent theorem.

\begin{theorem}\label{thm: wedge packing}%
  If $\classH$ is the class of all wedge packings,
  then the problem $\#\eiHom(\classH)$ is $\#\W$-hard. If $\cc{\#ETH}$ holds, it cannot be solved in time $f(k)\cdot n^{o(k/ \log k)}$.
\end{theorem}
\begin{proof}
  We reduce from the problem of counting $k$-matchings in bipartite
  graphs whose right-side vertices have degree $\leq 2$ and where any two
  distinct left-side vertices have at most one common neighbor.
  For this problem, Theorem~\ref{thm: match-restrbip} for bipartite graphs with girth greater 
  than~$4$ implies $\sw$-hardness and the desired bound under $\cc{\#ETH}$.
  Let $(G,k)$ be an instance of this problem, and let~$L(G)$ and~$R(G)$ be the
  left and right vertex sets, respectively.
  For $r\in\N$, we construct a graph $G^r$ as follows (see Figure~\ref{fig:gr}):

\begin{figure}[tbp]
\centering
\input{figures/graphic_gr.tex}
\input{figures/improved_graphic.tex}
\caption{\label{fig:gr} Example of the construction of $G^r$ as used in the proof of Theorem~\ref{thm: wedge packing}. The second row illustrates the correspondence between a $3$-matching in $G$ and the image of an edge-injective homomorphism from a wedge packing of size $3$ such that all wedges are \emph{good}. Furthermore we give examples for the image of a \emph{test} wedge and a \emph{bad} wedge.  }
\end{figure}
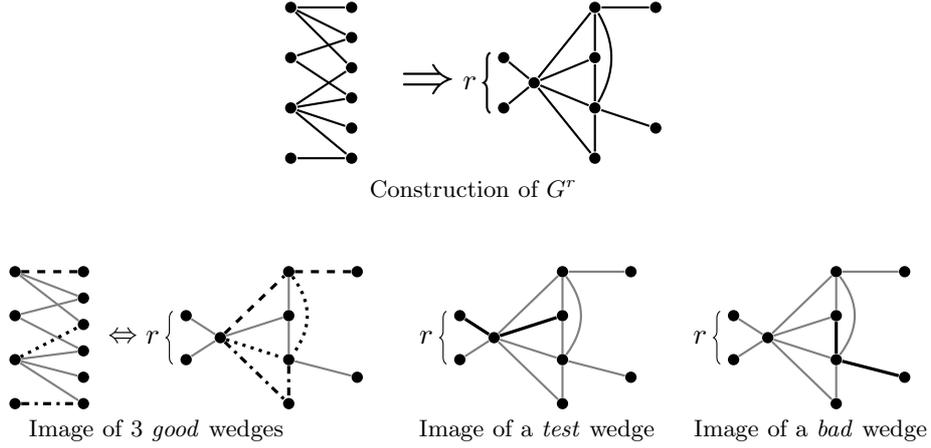
  \begin{enumerate}
    \item
      Insert a vertex $0$ that is adjacent to all vertices of $L(G)$.
    \item
      Add $r$ special vertices $1,\dots,r$ as well as the edges $01, 02, \dots,
      0r$.      
    \item
      For every vertex $v\in R(G)$ with $\deg(v)=2$, remove $v$ and add the set $N(v)$ as an edge to~$G^r$. Note that $|N(v)|=2$, so $N(v)$ can indeed be considered as an edge.

  \end{enumerate}
  Since $G$ is a simple graph and any two distinct vertices $u,v\in L(G)$ have
  at most one common neighbor in~$G$, the graph $G^r$ is again simple.
  Let $H=H_1 \dotcup \dots \dotcup H_k$ be the graph that consists of $k$
  vertex-disjoint copies of $P_2$.
  For $\varphi\in\eiHom(H, G^0)$, we say that a wedge $H_i$ is
  \begin{itemize}
    \item \emph{test} if $\varphi(H_i)$ contains two edges incident to $0$,
    \item \emph{good} if $\varphi(H_i)$ contains exactly one edge incident to
      $0$, and
    \item \emph{bad} if $\varphi(H_i)$ uses no edge incident to $0$.
  \end{itemize}
  Let $\alpha_{g,b}$ be the number of edge-injective homomorphisms $\varphi\in\eiHom(H,G^0)$ for which there are $0$ test wedges, $g$ good wedges, and $b$ bad wedges.
  \begin{claim} 
    \label{claim:alphak0}
    The number of $k$-matchings in $G$ is equal to $\alpha_{k,0} / \paren{2^k
      \cdot k!}$.
  \end{claim}
 
 \begin{claimproof}
    The integer $\alpha_{k,0}$ is the number of all $\varphi\in\eiHom(H,G^0)$
    such that the image of every~$H_i$ consists of a wedge that uses exactly
    one edge incident to~$0$.

    Given any such $\varphi$, we construct a $k$-matching~$M_\varphi$ of~$G$ as
    follows.
    For each~$i$, consider the wedge~$\varphi(H_i)$:
    It uses an edge $\set{0,v}$ for $v\in L(G)$ and an edge $\set{v,w}$ with
    $w\ne 0$.\linebreak
    If~$w\in R(G)$, then $N_G(w)=\set{v}$, and we add the edge $e_i$ with
    $e_i=\set{v,w}\in E(G)$ to the matching.
    Otherwise, we have $w\in L(G)$, and so the edge $\set{v,w}\in E(G^0)$
    corresponds to a vertex $u\in R(G)$ with $N_G(u)=\set{v,w}$, from which it
    was constructed.
    In this case, we add the edge $e_i$ with $e_i=\set{v,u}\in E(G)$ to the
    matching.
    Note that $e_i$ and $e_j$ for $i$ and $j$ with $i\ne j$ are disjoint; for if
    they shared a vertex $v\in L(G)$, the edge $\set{0,v}$ would be used by both
    $\varphi(H_i)$ and~$\varphi(H_j)$, and if they shared a vertex $v\in R(G)$,
    then either $N_G(v)$ or $N_G(v)\cup\set{v}$ would be an edge in $G^0$, which
    would be used by both $\varphi(H_i)$ and $\varphi(H_j)$.
    Thus the constructed set $M_\varphi$ is indeed a $k$-matching.

    On the other hand, for each $k$-matching~$M$, there are exactly $2^k\cdot k!$ edge-injective
    homomorphisms~$\varphi\in\eiHom(H,G^0)$ with $M=M_\varphi$ since the
    automorphism group of $H$ has this size.
  \end{claimproof}

  We aim at determining the number $\alpha_{k,0}$ by using
  an oracle for $\#\eiHom(\classH)$.
  Since we cannot directly ask the oracle to only count homomorphisms with a
  given number of bad and good wedges, we query the oracle multiple times and recover these numbers via a very specific form of interpolation fueled by Lemma~\ref{lem:interpolation}.
  To apply the lemma, we observe the following identity.
  \begin{claim} 
    \label{claim:betak}
    Let $k,r\in\N$. Then $\beta_k(G^r) := \#\eiHom(H,G^r)$ satisfies
    \begin{align*}
      \beta_k(G^r)
      &=
      \sum_{\substack{t,g,b\in\N \\t+g+b=k}}
      \alpha_{g,b}
      \cdot \binom{k}{g+b}
      \cdot
      (n+r-g)_{2t}
      \,.
    \end{align*}
  \end{claim}

  \begin{claimproof}
    We construct an element~$\varphi$ of
    $\eiHom(H_1\dotcup\dots\dotcup H_{k},G^r)$ whose image consists of $g$ good
    wedges, $b$ bad wedges, and $t$ test wedges, where $g+b+t=k$.
    There are $\binom{k}{g+b}$ possibilities to select the set of $H_i$ that
    will be mapped to a good or a bad wedge.
    Once this selection has been done, there are $\alpha_{g,b}$ edge-injective
    homomorphisms that map the selected~$H_i$ to $g$ good and $b$ bad wedges;
    to see this, note that $G^0$ and $G^r$ have exactly the same good and bad
    wedges.
    Finally, the test wedges can only be mapped to the edges incident to $0$,
    for which reason only the star with center~$0$ is relevant for the test
    wedges.
    Each good wedge that has already been placed blocks one edge of the star.
    Hence the~$t$ wedges map into a star $S_{n+r-g}$ with $n+r-g$ leaves.
    The number of edge-injective homomorphisms that map~$t$ wedges into a
    star with $\ell$ leaves is $(\ell)_{2t}$.
  \end{claimproof}
  Note that $\beta_k(G^r)$ is a polynomial in $r$ of degree at most $2k$.
  Setting $y=n+r$, Claim~\ref{claim:betak} yields a polynomial identity that is exactly 
  of the form required by Lemma~\ref{lem:interpolation}, and thus we can compute 
  the unknowns $\alpha_{g,b}$ for all $g,b\in\N$ with ${g+b\le k}$ from the 
  polynomials $\beta_0,\dots,\beta_{O(k)}$.
  Overall, the reduction runs in polynomial time, makes at most $O(k^2)$ queries
  to the oracle, and the parameter of each query is at most $O(k)$. This proves the $\sw$-hardness and the lower bound under $\cc{\#ETH}$.
\end{proof}

\section{Perfect matchings in line graphs of bipartite graphs}
\label{sec:perf line bipartite}

In this section, we prove Theorem~\ref{thm:perf_line}, 
which asserts that it is $\sharpP$-hard to count perfect matchings on line graphs of maximum degree $4$,
whereas this is polynomial-time solvable on $3$-regular line graphs.
We use a characterization of $3$-regular line graphs that was
established in~\cite{zhang2004disjoint}.
\begin{fact}
\label{fact: decomposeLine}
Every $3$-regular line graph $G$ with $\abs{V(G)} \geq 5$ is 
the union of two edge-disjoint graphs~$M$ and~$T$ on the same vertex set~$V(G)$,
where $M$ is a perfect matching and $T$ is a \emph{perfect triangle packing}.
That is,~$T$ is a vertex-disjoint union of triangles that covers all vertices of $G$.
Since all triangles of $G$ are contained in~$T$, the decomposition into~$M$ and~$T$ is unique.
\end{fact}

In the following, fix a graph $G$ with a partition into $M$ and $T$ as above,
and let $G_\downarrow$ be the graph obtained by contracting each triangle in $T$ to a single vertex without a self-loop.
Then~$G_\downarrow$ is~$3$-regular and it turns out that the perfect matchings of $G$ 
correspond bijectively to the odd edge-sets of $G_\downarrow$.
For the purposes of this section, we say that an edge-set $S \subseteq E(G_\downarrow)$ is \emph{odd} if,
for all $v\in V(G_\downarrow)$,
the degree of $v$ in the subgraph~$(V(G_\downarrow),S)$ is odd.
\begin{fact}
\label{fact: triangleG}
For all $t\in\set{0,\dots,|E(G_\downarrow)|}$, 
the odd edge-sets $S \subseteq E(G_\downarrow)$ of cardinality $t$ correspond bijectively to
the perfect matchings of $G$ that contain exactly $t$ edges from the matching $M$.
\end{fact}
\begin{proof}
  By construction of~$G_\downarrow$, every set~$S\subseteq E(G_\downarrow)$ corresponds to a set~${S'\subseteq E(M)}$ of the same size.
  Clearly~$S$ is odd if and only if~$S'$ is incident to exactly one or three vertices in each triangle of~$G$.
  In turn, the latter holds if and only if~$S'$ can be extended with edges of~$T$ to obtain a perfect matching of~$G$.
  Since the extension is unique if it exists, this defines a bijective mapping from odd edge-sets of~$G_\downarrow$ to perfect matchings of~$G$.
\end{proof}

Since odd edge-sets can be counted in polynomial time~\cite{DBLP:journals/iandc/CreignouH96},
we obtain the algorithmic result of Theorem~\ref{thm:perf_line} as follows.

\begin{lemma}
\label{lem: perfmatch-line-alg}
Given as input a $3$-regular line graph $G$,
the number of perfect matchings in $G$ can be computed in polynomial time.
\end{lemma}
\begin{proof}
If $|V(G)| < 5$, we apply brute-force.
Otherwise we decompose~$G$ into a perfect\linebreak matching~$M$ and a perfect triangle packing $T$ as in Fact~\ref{fact: decomposeLine}.
We obtain this decomposition in polynomial time by greedily removing triangles from~$G$;
as all triangles of $G$ are contained in the vertex-disjoint collection of triangles $T$,
this procedure does indeed recover $T$.

As a consequence of Fact~\ref{fact: triangleG}, 
the number of perfect matchings in $G$ equals the number of odd edge-sets in $G_\downarrow$.
By the algorithmic part of Theorem~4.4 in~\cite{DBLP:journals/iandc/CreignouH96},
counting odd edge-sets admits a polynomial-time algorithm.
\end{proof}

For the hardness result,
we reduce from counting perfect matchings in $3$-regular graphs,
which is $\sharpP$-hard \cite{Dagum.Luby1992}.
In the remainder, let $G$ be a $3$-regular graph.
Let~$G'$ be the graph obtained from $G$ by replacing every vertex $v\in V(G)$ by a triangle $T_v$
and attaching the $i$-th edge incident with $v$, for $i \in [3]$, to the $i$-th vertex of $T_v$.

By construction, the graph $G'$ admits a decomposition into a matching~$M$ and a triangle packing~$T$ as in Fact~\ref{fact: decomposeLine},
and we have $G = G'_\downarrow$.
Since the perfect matchings of any graph $A$ are precisely its odd edge-sets of cardinality $|V(A)|/2$,
we obtain the following corollary of Fact~\ref{fact: triangleG}:
\begin{fact}
\label{fact: perfmatchG-perfmatchLG}
The perfect matchings of $G$
correspond bijectively to
the perfect matchings of $G'$ that contain exactly $|V(G)|/2$ edges from the matching~$M$.
\end{fact}

In order to establish the reduction from counting perfect matchings in~$G$, it remains to count those perfect matchings of $G'$ that contain the desired number of edges from~$M$.
To this end, we replace each edge of~$M$ by a ``collar'' gadget from Figure~\ref{fig:collar_and_wire} so as to obtain a new graph $B$.

\begin{definition}
A \emph{collar} of length $\ell \in \mathbb N$ with end vertices $u,v$ is the graph obtained as follows:
Start with $\ell$ vertex-disjoint copies $C_1, \ldots, C_\ell$ of $K_4$.
For each $i$, let $a_i, b_i$ be two arbitrary distinct vertices of $C_i$.
For each $i$ with $1\le i < \ell$, add the edge $b_i a_{i+1}$.
Add a vertex $u$ and the edge $u a_1$,
and add a vertex $v$ and the edge $b_\ell v$.
\end{definition}
\begin{figure}
  \centering
\begin{tikzpicture}[-,thick,scale=0.6,main node/.style={circle,inner sep=1.5pt,fill}]

   \node[main node] (1) at (0,0) {};
   \node[main node] (2) at (2,0) {};
   \node[main node] (3) at (4,0) {};
   \node[main node] (4) at (6,0) {};
   \node[main node] (5) at (8,0) {};
   \node[main node] (6) at (10,0) {};
   \node[main node] (7) at (12,0) {};
   \node[main node] (8) at (14,0) {};
   \node[main node] (9) at (16,0) {};
   
   \draw (1) -- (2); \draw (2) -- (3); \draw (3) -- (4); \draw (4) -- (5); \draw (5) -- (6); \draw (6) -- (7); \draw (7) -- (8); \draw (8) -- (9);

   \node[main node] (30) at (3.5,1) {};
   \node[main node] (31) at (4.5,1) {};
   
   \node[main node] (50) at (7.5,1) {};
   \node[main node] (51) at (8.5,1) {};
   
   \node[main node] (70) at (11.5,1) {};
   \node[main node] (71) at (12.5,1) {};
   
   \draw (3) -- (30);
   \draw (3) -- (31);
   
   \draw (5) -- (50);
   \draw (5) -- (51);
   
   \draw (7) -- (70);
   \draw (7) -- (71);

   \node[inner sep=2.5pt,fill] (11) at (1,-3) {};
   \node[inner sep=2.5pt,fill] (12) at (3,-3) {};
   \node[inner sep=2.5pt,fill] (13) at (5,-3) {};
   \node[inner sep=2.5pt,fill] (14) at (7,-3) {};
   \node[inner sep=2.5pt,fill] (15) at (9,-3) {};
   \node[inner sep=2.5pt,fill] (16) at (11,-3) {};
   \node[inner sep=2.5pt,fill] (17) at (13,-3) {};
   \node[inner sep=2.5pt,fill] (18) at (15,-3) {};
   
   \draw (11) -- (12); \draw (12) -- (13); \draw (13) -- (14); \draw (14) -- (15); \draw (15) -- (16); \draw (16) -- (17); \draw (17) -- (18); 
   
   \node[inner sep=2.5pt,fill] (121) at (3.5,-2) {};
   \node[inner sep=2.5pt,fill] (131) at (4.5,-2) {};
   
   \draw (12) -- (121);
   \draw (12) -- (131);
   \draw (13) -- (121);
   \draw (13) -- (131);
   \draw (121) -- (131);
   
   \node[inner sep=2.5pt,fill] (141) at (7.5,-2) {};
   \node[inner sep=2.5pt,fill] (151) at (8.5,-2) {};
   
   \draw (14) -- (141);
   \draw (14) -- (151);
   \draw (15) -- (141);
   \draw (15) -- (151);
   \draw (141) -- (151);
   
   \node[inner sep=2.5pt,fill] (161) at (11.5,-2) {};
   \node[inner sep=2.5pt,fill] (171) at (12.5,-2) {};
   
   \draw (16) -- (161);
   \draw (16) -- (171);
   \draw (17) -- (161);
   \draw (17) -- (171);
   \draw (161) -- (171);
\end{tikzpicture}
\caption{\label{fig:collar_and_wire}A barbed wire of length $3$ (\emph{top}) and a collar of length $3$ (\emph{bottom}).}
\end{figure}
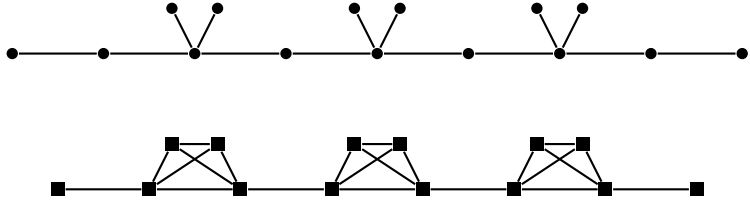

In the graph~$B$, each collar gadget intersects the remainder of $B$ only in its end vertices.
Perfect matchings of~$B$ may block zero, one, or both ends of a collar by edges outside of the collar.
In each case, there is a simple closed expression for the number of perfect matchings in the remaining part of the collar.

\begin{fact}
\label{fact: collarmatch}
A collar $X_\ell$ of length $\ell \in \mathbb N$ with end vertices $u,v$ has exactly one perfect matching.
(It contains the two edges incident to the ends, the~$\ell-1$ edges between $K_4$-copies and, for each $i$,
the unique edge in $C_i$ that is disjoint from $a_ib_i$.)
The graphs $X_\ell - u$ and $X_\ell - v$ have an odd number of vertices and thus no perfect matching.
Finally, the graph $X_\ell - \{u,v\}$ has exactly~$3^\ell$ perfect matchings, as we can independently choose one of three perfect matchings in each $K_4$-copy.
\end{fact}

We can now prove the hardness result in Theorem~\ref{thm:perf_line}.

\begin{lemma}
\label{lem: perfmatch-line-hardness}
Counting perfect matchings in line graphs of bipartite graphs is $\sharpP$-hard,
even when the input graph has maximum degree $4$.
\end{lemma}
\begin{proof}
Let $G$ be a $3$-regular graph for which we wish to determine the number of perfect matchings.
Recall that $G'$ is obtained by inserting triangles at vertices
and that $G'$ decomposes into a perfect matching~$M$ and a perfect triangle packing~$T$.
By Fact~\ref{fact: perfmatchG-perfmatchLG}, 
the number of perfect matchings of~$G$ is equal to the number of perfect matchings of~$G'$ that contain exactly $\abs{V(G)} / 2$ edges from the matching~$M$.

Let $B$ be obtained from $G'$ by replacing each edge $uv \in M$
with a fresh collar of length $|E(G')|+1$ with ends $u$ and $v$.
Write 
\[
R := 3^{|E(G')|+1}.
\]
For $0 \leq t \leq |E(M)|$, 
let $m_t$ be the number of perfect matchings of~$G'$ that contain exactly~$t$ edges from~$M$.
By Fact~\ref{fact: collarmatch}, the number of perfect matchings in $B$ satisfies
\begin{equation}\label{eq: perfmach m_t R^t}
  \#\mathrm{PerfMatch}(B) = \sum_{t=0}^{\abs{E(M)}} m_t \cdot R^{\abs{E(M)}-t}\,.
\end{equation}

Clearly $m_t<R$ holds for all $t \in \mathbb N$.
Thus \eqref{eq: perfmach m_t R^t} can be viewed as a representation of the integer
$\#\mathrm{PerfMatch}(B)$ in base $R$, 
and then the $(\abs{E(M)}-t)$-th digit in this representation is precisely~$m_t$.
Given the value of $\#\mathrm{PerfMatch}(B)$,
the values of~$m_t$ are uniquely determined
and can be recovered by elementary arithmetic.
This way, we obtain $m_{|V(G)|/2}$, the number we wish to compute.

It is clear~$B$ has maximum degree~$4$, since~$G'$ is $3$-regular 
and replacing edges with collars increases the maximum degree to~$4$.
It remains to prove that $B$ is the line graph of a bipartite graph.
To this end, we construct a bipartite graph $S$ with $B = L(S)$ as follows:
Starting from~$G$, replace each edge $uv \in E(G)$ by a \emph{barbed wire} of length $\ell = \abs{E(G')} + 1$.
This is a $(u,v)$-path with $2 \ell + 2$ edges, 
in which each of the~$\ell$ odd-numbered internal vertices has two additional leaf-edges attached
  (see Figure~\ref{fig:collar_and_wire} for an example with $\ell = 3$.)
The line graph of a barbed wire of length $\ell \in \mathbb N$ is 
a collar of length $\ell \in \mathbb N$.
From this fact, it can be verified that $B$ is the line graph of $S$.
\end{proof}

Together, Lemmas~\ref{lem: perfmatch-line-alg} and \ref{lem: perfmatch-line-hardness} prove Theorem~\ref{thm:perf_line}.

\section{Edge-injective homomorphisms}\label{sec:homs}

In this section we prove Theorem~\ref{thm:dichotomy}, our 
complexity dichotomy theorem for counting edge-injective homomorphisms.
Recall that $\#\eiHom(H,G)$ is the number of edge-injective homomorphisms from~$H$ to~$G$.
A set $S\subseteq V(H)$ is a \emph{weak vertex-cover} if every edge $e\in E(H)$
either has a non-empty intersection with $S$ or $e$ does not have any other edges 
incident to it.
The \emph{weak vertex-cover number} of~$G$ is the minimum size of a weak 
vertex-cover of~$G$.
A family of graphs $\classH$ has \emph{bounded} weak vertex-cover number if this number can be uniformly bounded by a constant $c=c(\classH)$ for all graphs $H \in \classH$; otherwise this number is \emph{unbounded} for~$\classH$.

\subsection{Polynomial-time algorithm for bounded weak vertex-cover number}

In this section, we present a polynomial-time algorithm for counting edge-injective homomorphisms from pattern graphs with bounded weak vertex-cover number. 
To improve readability, we drop the cardinality signifier~$\#$ in this section.
\begin{theorem} 
  \label{thm:polytime}
  For any constant $c\in\N$,
  given a graph~$H$ with a weak vertex-cover of size at most~$c$ and a graph~$G$, 
  we can compute~$\eiHom(H,G)$ in time $O(n^{c'})$,
  where $c'$ is a constant depending only on $c$.
\end{theorem}

In a preprocessing step, our algorithm removes isolated vertices and edges by exhaustively applying the following reduction rules.
\begin{lemma}[Deleting isolated vertices and edges]\label{lem: reduce isolated}
  Let $H$ and $G$ be graphs.
  \begin{itemize}
    \item
      If $v$ is an isolated vertex in~$H$ and $H-v$ is obtained by deleting~$v$ from~$H$, then
      \[\eiHom(H,G)=\abs{V(G)}\cdot\eiHom(H-v,G)\,.\]
    \item
      If~$e=\set{u,v}$ is an isolated edge in~$H$, then
      \[\eiHom(H,G)=2\paren{\abs{E(G)}-\abs{E(H)}+1}\cdot \eiHom(H-u-v,G)\,.\]
  \end{itemize}
\end{lemma}
\begin{proof}
  The first item is trivial. To prove the second item, first note that every edge-injective homomorphism~$h$ from $H-e$ to~$G$ has exactly $\abs{E(H)}-1$ edges of~$G$ in its image.
  To extend~$h$ to an edge-injective homomorphism from~$H$ to~$G$, we have to map~$e$ to an edge that is distinct (but not necessarily disjoint) from the edges in the image of~$h$.
  There are $\abs{E(G)}-\abs{E(H)}+1$ candidates for the image of~$e$, and once an image $e'$ of $e$ has been determined, we can independently choose one of the two orientations to map $u$ and $v$ to the endpoints of $e'$.
  Since every edge-injective homomorphism from~$H$ to~$G$ is obtained in this way exactly once, the claim follows.
\end{proof}

By this preprocessing,
we can assume that $H$ contains neither isolated vertices nor isolated edges. 
We then reduce the counting of edge-injective homomorphisms to the counting of embeddings.
We achieve this by writing $\eiHom(H,.)$ as a linear combination of embedding numbers $\Emb(F,.)$ for suitable graphs~$F$.
Let $\partitions H$ be the set of all partitions~$\rho$ of the vertex set~$V(H)$. 
Given a partition $\rho \in \partitions H$, let~$\contract H\rho$ be the quotient graph of~$H$,
which is obtained by merging each block of~$\rho$ into a vertex.
\begin{lemma}\label{lem: eiHom graph motif}
  Let $H$ and~$G$ be graphs.
  Call a partition~$\rho\in\partitions H$ \emph{edge-injective} if,
  for every pair of blocks~$B,B'\in\rho$, 
  there is at most one edge between $B$ and $B'$ in~$H$.
  Then
  \begin{equation}\label{eq: EdgInj = Lincomb Sub}
    \eiHom(H,G)
    =
    \sum_{\substack{\rho\in\partitions{H} \\ \textrm{edge-injective}}} \Emb(\contract H\rho,G)\,.
  \end{equation}
\end{lemma}
\begin{proof}
  There exists a bijection between edge-injective homomorphisms~$h$ from~$H$ to~$G$ and pairs~$(\rho,g)$ where~$\rho\in\partitions H$ is edge-injective and~$g$ is an injective homomorphism from~$\contract H\rho$ to~$G$.
  To define~$(\rho,g)$ from~$h$, put vertices of~$H$ into the same block~$B$ of~$\rho$ if and only if they map to the same vertex~$v\in V(G)$ under~$h$; then~$g$ maps that block~$B$, which is a vertex in~$\contract H\rho$, to~$v$.
  Conversely, if $(\rho,g)$ is a given pair, the canonical homomorphism~$f$ that maps~$H$ to~$\contract H\rho$ is edge-injective, and we set~$h=g\circ f$.
  It is easy to check that this is indeed the required bijection.
\end{proof}

Each quantity $\Emb(\contract H\rho,G)$ on the right side of \eqref{eq: EdgInj = Lincomb Sub} can be computed in time~$n^{\vc{\contract{H}\rho}+O(1)}$ by known algorithms~\cite{williams2013finding,curticapean2014complexity},
where $\vc{\contract{H}\rho}$ is the vertex-cover number of $\contract{H}\rho$.
Since every vertex-cover of $H$ is also a vertex-cover of $\contract{H}\rho$, we can combine the preprocessing rules for isolated vertices and edges with~\eqref{eq: EdgInj = Lincomb Sub} to obtain an $n^{c+O(1)}$ time algorithm for~$\eiHom(H,.)$ for any fixed graph~$H$ of weak vertex-cover number~$c$.
However, when taking the size of $H$ into account as $k=\abs{V(H)}$, 
there are up to $k^{\Omega(k)}$ quotient graphs~$\contract H\rho$.
Thus, for patterns $H$ of constant weak vertex-cover number,
the algorithm via~\eqref{eq: EdgInj = Lincomb Sub} is fixed-parameter tractable in the parameter~$k$, 
but it does not run in polynomial time in $k$.

To obtain a polynomial-time algorithm and thus prove Theorem~\ref{thm:polytime},
we collect terms for isomorphic quotients in \eqref{eq: EdgInj = Lincomb Sub}
in such a way that the resulting reduced linear combination has polynomial length 
and its coefficients can be computed in polynomial time.
More precisely, we define an equivalence relation on~$\partitions{H}$ that has only polynomially many equivalence classes and we show that the size of each equivalence class can be computed efficiently.
The definition of this equivalence relation is somewhat technical, so we interleave it with a discussion of a running example in Figure~\ref{fig:gralgo} to facilitate reading.
\begin{figure}[ptb]
\centering
\input{figures/graphic_algo.tex}
\caption{\label{fig:gralgo}
The figure shows a graph $H$ with a vertex cover $C=\set{u,v,w,x}$, depicted as squares.
There is a partition $\rho_C=\set{\set{u},\set{v,x},\set{w}}$ of the vertices in the vertex cover, and $\rho_1$, $\rho_2$, and $\rho_3$ are (edge-injective) extensions of $\rho_C$ such that the resulting graphs are isomorphic.
The color allocations $\coloralloc_{\rho_1}$ and $\coloralloc_{\rho_2}$ are equal and hence $\rho_1$ and $\rho_2$ are equivalent.
On the other hand, the color allocations $\coloralloc_{\rho_1}$ and $\coloralloc_{\rho_3}$ are not equal and hence $\rho_1$ and $\rho_3$ are not equivalent.
}
\end{figure}
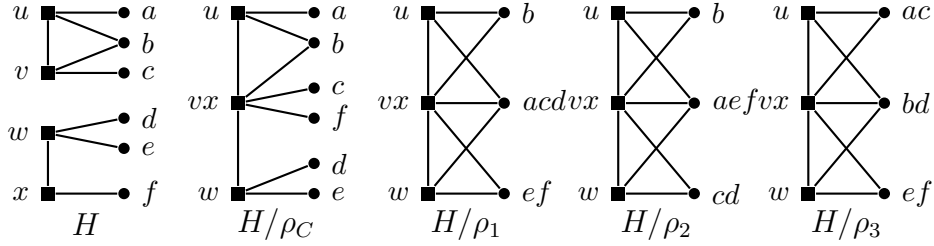
\begin{definition}
  Let~$H$ be a graph with a fixed vertex-cover~$C$.
  Let $\rho\in\partitions{H}$ be an edge-injective partition.

  \begin{enumerate}[leftmargin=2.5em,label=(\roman*)]
    \item
      The \emph{vertex-cover sub-partition of~$\rho$} is the set $\rho_C\subseteq\rho$ of all blocks~$B\in\rho$ that intersect~$C$.
      We write $\contract H {\rho_C}$ for the graph obtained from $H$
      by contracting each block $B \in \rho_C$ to a single vertex,
      and we speak of the resulting vertices as the \emph{$C$-image vertices of $\rho$}.
      We define $I := V(H)\setminus\bigcup\rho_C$; 
      these are the vertices of $H$ not covered by $\rho_C$.
  \end{enumerate}
  All three partitions~$\rho_1,\rho_2,\rho_3\in\partitions{H}$ in Figure~\ref{fig:gralgo} have the same vertex-cover sub-partition~$\rho_C$.
  The $C$-image vertices of these partitions are shown on the left sides of the graphs.
  To simplify the figure, we have chosen the partitions in such a way that,
  for any $\rho \in \{\rho_1,\rho_2,\rho_3\}$,
  each block in $\rho_C$ is in fact fully contained in $C$.
  That is, $\rho_C$ is a partition of $C$.
  In general however,~$\rho_C$ could also contract vertices from $V(H)\setminus C$ into $C$.

  \begin{enumerate}[resume*]
    \item
      Sets~$K\subseteq\rho_C$ are called \emph{colors}.
      These are the possible neighborhoods of vertices in~$I$.
      That is, each vertex~$v\in I$ has a \emph{color} $K(v)\subseteq\rho_C$, which is the set of blocks~$B\in\rho_C$ that~$v$ is adjacent to in~$H$.
      (Here, we say that $v$ is adjacent to a block $B$ if $v$ is adjacent to\linebreak some $w \in B$ in $H$.)
  \end{enumerate}
  In Figure~\ref{fig:gralgo}, the colors of vertices in $I$ can be recognized well in the graph labeled with~$\contract H{\rho_C}$:
  The color of a vertex in $I$ is its neighborhood in~$\contract H{\rho_C}$,
  and this neighborhood is fully contained among the $C$-image vertices of $\rho$.
  For example, the color of~$a$ is~$\set{\set{u}}$, and the color of~$b$ is~$\set{\set{u},\set{v,x}}$.
  \begin{enumerate}[resume*]
    \item
      Consider the set $\rho\setminus\rho_C$;
      this is a partition of the set $I$.
      When collecting for isomorphic quotient graphs,
      we will not need know $\rho\setminus\rho_C$ in its entirety.
      Rather, it is sufficient to know how many vertices of each color were identified by $\rho\setminus\rho_C$.
      To this end, we define the \emph{color allocation~$\coloralloc_\rho$ of~$\rho$} is the multiset
      \[\coloralloc_\rho := \msetc{\setc{K(v)}{v\in B} \ }{\ B\in\rho\setminus\rho_C}.\] 
      That is, for each block~$B\in\rho\setminus\rho_C$, the color allocation contains a copy of the set of all colors appearing in~$B$.
  \end{enumerate}
  In Figure~\ref{fig:gralgo}, the color allocations of~$\rho_1$, $\rho_2$, and $\rho_3$ satisfy
  \begin{align*}
    \coloralloc_{\rho_1}
    &=
    \Big\{\!\!\!\Big\{
    {\set{K(b)},\set{K(a),K(c),K(d)},\set{K(e),K(f)}}
    \Big\}\!\!\!\Big\}
    \,,\\
    \coloralloc_{\rho_2}
    &=
    \Big\{\!\!\!\Big\{
    {\set{K(b)},\set{K(a),K(e),K(f)},\set{K(d),K(c)}}
    \Big\}\!\!\!\Big\}
    \,,\\
    \coloralloc_{\rho_3}
    &=
    \Big\{\!\!\!\Big\{
    {\set{K(a),K(c)},\set{K(b),K(d)},\set{K(e),K(f)}}
    \Big\}\!\!\!\Big\}
    \,.
  \end{align*}
  Since~$K(e)=K(d)$ and $K(c)=K(f)$, the color allocations for~$\rho_1$ and~$\rho_2$ are actually equal.
  On the other hand, $\rho_1$ and $\rho_3$ have distinct color allocations,
  since $\rho_1$ and $\rho_3$ cannot be obtained from each other by swapping vertices of the same color for each other.

  If two partitions of $V(H)$ have the same vertex-cover sub-partition and the same color-allocation,
  then we will consider them to be equivalent for our purpose of collecting isomorphic quotient graphs.
  \begin{enumerate}[resume*]
    \item
      Two partitions~$\rho,\rho'\in\partitions{H}$ are called \emph{equivalent} if
      $\rho_C=\rho'_C$ and
      ${\coloralloc_\rho=\coloralloc_{\rho'}}$ hold.
      This defines an equivalence relation on~$\partitions{H}$; the equivalence class of~$\rho$ is uniquely determined by the pair~$(\rho_C,\coloralloc)$.
      We write $\rho_{\rho_C,\coloralloc}$ for an arbitrary fixed representative of this class.
  \end{enumerate}

In the following, we collect some properties on the interplay of the above definitions
that we will ultimately put to use to obtain a polynomial-time algorithm.

\end{definition}
\begin{lemma}
\label{lem:equivalence properties}
  Let $H$ be a graph without isolated vertices, let~$C$ be a vertex-cover of~$H$, and let~$k=\abs{V(H)}$.
  Let $\rho \in \partitions H$ be an edge-injective partition.
  \begin{enumerate}[leftmargin=2.5em,label=(\roman*)]
    \item\label{item: size bound}
      We have $\abs{\rho_C}\le \abs{C}$,
      that is, there are at most $|C|$ vertices that are $C$-image vertices of $\rho$.
      We also have $\abs{\bigcup\rho_C}\le \abs{C}^2$,
      that is, at most $\abs{C}^2$ vertices are contracted to obtain the $C$-image vertices of $\rho$.
    \item\label{item: betas}
    Recall the color allocation $\coloralloc_\rho$.
      Every element $\beta \in \coloralloc_\rho$ is a partition of the set~$\bigcup\beta$, 
      which is a set of colors.
      Writing~$\coloralloc_\rho(\beta)$ for the number of times~$\beta$ occurs as an element in the multiset~$\coloralloc_\rho$,
      we have $\coloralloc_\rho(\beta) \leq k$.
    \item\label{item: isomorphic}
      If $\rho' \in \partitions H$ is edge-injective and equivalent to~$\rho$, 
      as defined above, then the quotient graphs $\contract H\rho$ and $\contract H{\rho'}$ are isomorphic.
    \item\label{item: multinomial}
      Consider the equivalence class of partitions in $\partitions H$ corresponding to the pair $(\rho_C,\coloralloc_\rho)$.
      The size of this equivalence class is
    \begin{align*}
    N(\rho_C,\coloralloc_\rho)
    :=
    \paren*{
    	\prod_{K\subseteq\rho_C}
    	\binom{k_K}{\coloralloc_\rho(\beta_1),\dots,\coloralloc_\rho(\beta_\ell)}
    }
    \cdot
    \paren*{
    	\prod_{\beta}
    	\left(\coloralloc_\rho(\beta)!\right)^{\abs{\beta}-1}
    }.
    \end{align*}
    Here, $\beta_1,\dots,\beta_\ell$ is an enumeration of all sets~$\beta\in\coloralloc_\rho$ with~$K\in\beta$,
    and~$k_K$ is the number of vertices in~$I=V(H)\setminus\bigcup\rho_C$ that have color~$K$.
  \end{enumerate}
\end{lemma}
\begin{proof}
  \ref{item: size bound}
  The first inequality follows directly from the definition.
  For the second inequality,\linebreak let $D=\bigcup\rho_C\subseteq V(H)$ be the set of all vertices occurring in~$\rho_C$; 
  this is exactly the set of vertices that get merged into $C$ in $\contract H\rho$.
  For edge-injective partitions~$\rho$, the graph $H[D]$ contains at most $\binom{\abs{\rho_C}}{2}$ edges.
  Furthermore, since~$\abs{\rho_C}\le\abs{C}$ 
  and since every vertex in~$D\setminus C$ is adjacent to~$C$, 
  we have~$\abs{D}\le \abs{C}^2$.

  \ref{item: betas}
  Let $\beta\in\coloralloc_\rho$.
  By definition there is a block~$B\in\rho\setminus\rho_C$ with $\beta=\setc{K(v)}{v\in B}$.
  Recall that~$K(v)$ is the set of blocks of~$\rho_C$ (or equivalently, the set of $C$-image vertices of $\rho$) adjacent to~$v$ in the graph~$H$.
  Since~$\rho$ is edge-injective, any two distinct vertices~$v,v'\in B$ have~$K(v)\cap K(v')=\emptyset$.
  Thus $\beta$ is indeed a partition of the set~$\bigcup\beta=\bigcup_{v\in B} K(v)$.
  
  \ref{item: isomorphic}
  Let $\rho$ and $\rho'$ be equivalent.
  Then $\rho_C=\rho'_C$ holds, and so the quotients~$\contract H\rho$ and~${\contract H{\rho'}}$ are identical on the vertex set~$\rho_C$.
  Moreover, $\rho_C$ is a vertex cover for both quotient graphs, and so the quotient graphs are isomorphic if and only if, for each set~$K\subseteq\rho_C$, the two graphs have the same number of vertices~$v$ outside of~$\rho_C$ with $N(v)=K$.
  We prove that this is the case.
  
  Let~$K\subseteq\rho_C$ be a color.
  We want to prove that $\contract H\rho$ and~$\contract H{\rho'}$ have the same number of vertices outside of $\rho_C$
  whose neighborhood is exactly~$K$.
  Recall that~$\coloralloc_\rho(\beta)$ is equal to the number of blocks~$B\in\rho\setminus\rho_C$ with~${\setc{K(v)}{v\in B}} = \beta$.
  Note that~$B$ is a vertex in the quotient graph~$\contract H\rho$ with neighborhood~${\bigcup\beta\subseteq\rho_C}$.
  As a result, the value $\sum_{\beta\in\partitions{K}} \coloralloc_\rho(\beta)$ is equal to the number of blocks~${B\in\rho\setminus\rho_C}$ such that the vertex~$B$ has neighborhood exactly~$K$ in the quotient graph~$\contract H\rho$.
  Since these values are equal for $\rho$ and~$\rho'$, the claim follows, and the two quotient graphs are isomorphic.
  
  \ref{item: multinomial}
  To see this, consider a color $K \subseteq \rho_C$.
  By definition, there are exactly~$k_K$ vertices in~$I = V(H)\setminus\bigcup\rho_C$ that have color~$K$.
  Let us call this set~$V_K$.
  Furthermore let $\beta_1,\dots,\beta_\ell$ be an enumeration of all sets in~$\coloralloc_\rho$ that contain $K$.
  All partitions~$\rho'$ that are equivalent to~$\rho$ use~$\coloralloc_\rho(\beta_1)$ vertices from~$V_K$ to form vertices with neighborhood~$\beta_1$ in~$\contract H\rho$.
  There are $\binom{k_K}{\coloralloc_\rho(\beta_1)}$ choices for these vertices.
  Similarly, there are $\binom{k_K - \coloralloc_\rho(\beta_1)}{\coloralloc_\rho(\beta_2)}$ choices for selecting those vertices from~$V_K$ that form vertices with neighborhood~$\beta_2$ in~$\contract H\rho$.
  Continuing this process inductively yields the claimed multinomial coefficient for the number of possibilities to allocate the vertices of~$V_K$ to the~$\beta_i$.
  Since these choices are independent for independent colors~$K$, multiplying the multinomial coefficients yields the total number of possibilities to allocate the vertices of~$I$ to the sets~$\beta\in\coloralloc_\rho$.

  Once vertices of~$I$ have been allocated to the~$\beta$, we are still free to choose,
  for each fixed~$\beta$,
  which vertices to put in the same block of~$\rho'\setminus\rho_C$.
  We assembled sets $S_1,\dots,S_{\abs{\beta}}\subseteq I$, each of size $\coloralloc_\rho(\beta)$. Each $S_i$ contains vertices whose color is the~$i$-th color of $\beta$. We need to construct~$\coloralloc_\rho(\beta)$ blocks in~$\rho'\setminus\rho_C$, such that each block contains exactly one element from every $S_i$. For the first block, we have $\coloralloc_\rho(\beta)^{\abs{\beta}}$ choices. Once the first block is fixed, we have $(\coloralloc_\rho(\beta)-1)^{\abs{\beta}}$ choices for the second one, and so on. Thus the overall number of choices to enumerate the elements of~$\rho'\setminus\rho_C$ is $\left(\coloralloc_\rho(\beta)!\right)^{\abs{\beta}}$. Finally, since $\rho'\setminus\rho_C$ is a set and does not care about the order of the blocks, we divide by $\coloralloc_\rho(\beta)!$. This results in $\left(\coloralloc_\rho(\beta)!\right)^{\abs{\beta}-1}$ choices to construct the set of blocks for~$\beta$. These choices are independent for different~$\beta$, so their product yields the number of~$\rho$ that are consistent with~$\rho_C$ and~$\coloralloc_\rho$.
\end{proof}
\begin{lemma}[Collecting terms]
  Let $H$ and~$G$ be graphs such that~$H$ has no isolated vertices.
  Let $C$ be a fixed vertex-cover of~$H$.
  Then we have
  \begin{equation}\label{eq: EdgInj = Lincomb Sub again}
    \eiHom(H,G)
    =
    \sum_{(\rho_C,\coloralloc)}
    N(\rho_C,\coloralloc)
    \cdot
    \Emb\paren[\Big]{\contract H{\paren{\rho_{\rho_C,\coloralloc}}},G}\,,
  \end{equation}
  where the sum is over all equivalence classes~$(\rho_C,\coloralloc)$ of~$\partitions H$.
\end{lemma}
\begin{proof}
  We start with~\eqref{eq: EdgInj = Lincomb Sub} in Lemma~\ref{lem: eiHom graph motif} and collect terms for equivalent~$\rho\in\partitions H$.
  Since the collected terms lead to isomorphic quotient graphs by Lemma~\ref{lem:equivalence properties}\ref{item: isomorphic}, the numbers~$\Emb(\contract H\rho,G)$ are identical in each equivalence class, and they are equal to~$\Emb(\contract H{\rho_{\rho_C,\coloralloc_\rho}},G)$.
  The number of collected terms for each equivalence class is equal to~$N(\rho_C,\coloralloc_\rho)$ by
  Lemma~\ref{lem:equivalence properties}\ref{item: multinomial}.
  This implies~\eqref{eq: EdgInj = Lincomb Sub again}.
\end{proof}
\begin{proof}[Proof of Theorem~\ref{thm:polytime}]
  The following algorithm computes~$\eiHom(H,G)$ via~\eqref{eq: EdgInj = Lincomb Sub again}.

  \begin{algor}{A}{$\eiHom$}{Given~$H$ and~$G$, this algorithm computes $\eiHom(H,G)$.}
  \item[A1]
    Exhaustively apply the reduction rules from Lemma~\ref{lem: reduce isolated}.
    [Now~$H$ does not have isolated vertices or edges.]
  \item[A2]
    Compute a minimum vertex-cover~$C$ of~$H$ via exhaustive search.
  \item[A3]
    Iterate over all equivalence classes of $\partitions H$;
    this can be achieved by iterating over pairs~$(\rho_C,\coloralloc)$
    where $\rho_C$ is a valid vertex-cover sub-partition and $\coloralloc$ is a valid color allocation.
    For each equivalence class:
    \begin{itemize}[labelsep=1ex]
      \item Compute $N(\rho_C,\coloralloc)$.
      \item Query the oracle for the value
        $\Emb(\contract H{\rho_{\rho_C,\coloralloc}},G)$.
    \end{itemize}
  \item[A4] Output the sum on the right side of~\eqref{eq: EdgInj = Lincomb Sub again}.
  \end{algor}
  Clearly, the steps A1 and A4 take polynomial time.
  Step A2 takes polynomial time, since we assumed~$\abs{C}\le c\le O(1)$.
  Moreover, the number~$N(\rho_C,\coloralloc)$ and the graph~$\contract H{\rho_{\rho_C,\coloralloc}}$ can be computed in polynomial time.
  In place of the oracle for~$\Emb$, we use the known $n^{\vc H+O(1)}$ time algorithm~\cite{williams2013finding,curticapean2014complexity}.
  It remains to clarify how to iterate over the equivalence classes in~A3, why there are not too many of them, and how to compute a representative~$\rho_{\rho_C,\coloralloc}$.

  To iterate over all candidates for~$\rho_C$, note that $\rho_C$ is a partition of~$\bigcup\rho_C$ and satisfies the bound~$\abs{\bigcup\rho_C} \le c^2$ from Lemma~\ref{lem:equivalence properties}\ref{item: size bound}.
  Thus our algorithm exhaustively iterates over all sets~$D\subseteq V(H)$ with~$\abs{D}\le c^2$.
  For each~$D$, it exhaustively iterates over all~$\rho_C\in\partitions{D}$ with the property that every block of~$\rho_C$ intersects~$C$.
  Let~$k=\abs{V(H)}$.
  The number of candidates for~$\rho_C$ is bounded by~$\binom{k}{\le c^2}\cdot B_{c^2}$, where the first factor reflects the possible choices for~$D$ 
  and the second is the \emph{Bell number}, 
  the number of partitions of a $c^2$-element set.
  Thus there are only polynomially many choices for the candidates of~$\rho_C$.

  By Lemma~\ref{lem:equivalence properties}\ref{item: betas}, all color allocations~$\coloralloc$ are multisets that contain at most~$k$ duplicates of each member.
  Moreover, each member of~$\coloralloc$ is a partition~$\beta\in\partitions{K}$ of some set~$K\subseteq\rho_C$.
  Since there are at most~$2^c$ subsets~$K$, each with at most~$B_{c}$ partitions~$\beta$, the number of \emph{distinct} elements of~$\coloralloc$ is at most~$2^c B_c\le c'$ for some large enough constant~$c'$.
  Each of these elements can occur between~$0$ and~$k$ times in $\coloralloc$, so once a candidate for~$\rho_C$ has been fixed, the number of distinct candidates for~$\coloralloc$ is bounded by~$k^{c'}$, which is a polynomial.

  We conclude that A3 can be executed in a polynomial number of iterations.
  The candidates for~$\rho_C$ are partitions of size-$(\le c^2)$ subsets of~$V(H)$ and the candidates for~$\coloralloc$ are multisets of partitions of subsets of~$\rho_C$.
  Once a candidate~$(\rho_C,\coloralloc)$ has been fixed, it remains to argue that we can construct a representative~$\rho_{\rho_C,\coloralloc}$ if it exists.

  Given $(\rho_C,\coloralloc)$, we construct~$\rho\supseteq\rho_C$ as follows.
  For each~$K\subseteq\rho_C$ and each~$\beta\in\partitions{K}$, we do the following $\coloralloc(\beta)$ times:
  Pick arbitrary vertices~$v_1,\dots,v_{\abs{\beta}}\in V(H)\setminus\bigcup\rho_C$
  such\linebreak that~$\beta=\setc{K(v_i)}{1\le i\le \abs{\beta}}$, add the set~$\set{v_1,\dots,v_{\abs{\beta}}}$ to~$\rho$, and mark the vertices as used.
  If we run out of vertices in~$V(H)\setminus\bigcup\rho_C$ when doing so, or vertices are left unused at the end, then there is no~$\rho\in\partitions{H}$ with vertex-cover sub-partition~$\rho_C$ and color allocation~$\coloralloc$ (and thus $N(\rho_C,\coloralloc)=0$ holds for this candidate, since it did not represent an equivalence class).
  Otherwise we have constructed a partition~$\rho\in\partitions{H}$.
  If~$\rho$ is edge-injective, we output it.
  Otherwise, we again have~$N(\rho_C,\coloralloc)=0$ (in this case, the candidate~$\rho_C$ was not edge-injective to begin with).
\end{proof}

\subsection{Hardness for hereditary graph classes}

We now consider graph classes~$\classH$ that do \emph{not} have bounded weak 
vertex-cover number, and we prove that $\#\eiHom(\classH)$ is 
$\sw$-complete if $\classH$ has the additional property of being hereditary.
To this end, we first show that every graph class of unbounded weak vertex-cover number contains one of the six basic graph classes depicted in 
Figure~\ref{fig:IndSubs} as induced subgraphs.

For the purposes of this paper, we say that a graph is a windmill $W_k$ of size~$k$ if
it is a matching of size $k$ with an additional \emph{center vertex} adjacent to every other vertex.
Moreover, the \emph{subdivided star~$SS_k$} is a $k$-matching with a center 
vertex that is adjacent to exactly one vertex of each edge in the matching.
A \emph{triangle packing $k\cdot K_3$} is the disjoint union of $k$ triangles, a 
\emph{wedge} is a path~$P_2$ that consists of two edges,
and a \emph{wedge packing $k\cdot P_2$} is the disjoint union of $k$ wedges.

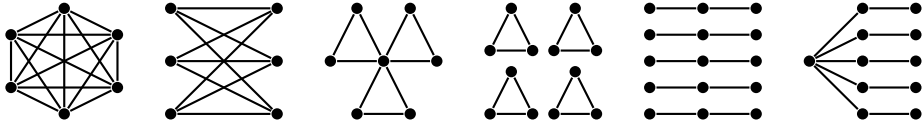
\begin{figure}[tp]
  \centering
    \begin{tikzpicture}[-,thick, scale = 0.7]
		
		\node[circle,inner sep=1.5pt,fill](1) at (0,0.5) {};
		\node[circle,inner sep=1.5pt,fill](2) at (0,1.5) {};
		\node[circle,inner sep=1.5pt,fill](3) at (1,2) {};
		\node[circle,inner sep=1.5pt,fill](4) at (1,0) {};
		\node[circle,inner sep=1.5pt,fill](1337) at (2,0.5) {};
		\node[circle,inner sep=1.5pt,fill](7331) at (2,1.5) {};
		\draw (1) -- (2); \draw (1) -- (3); \draw (1) -- (4);
		\draw (2) -- (3); \draw (2) -- (4); \draw (3) -- (4);
		\draw (1337) -- (7331); \draw (1337) -- (1);
		\draw (1337) -- (2); \draw (1337) -- (3); \draw (1337) -- (4);
		\draw (7331) -- (1); \draw (7331) -- (2); \draw (7331) -- (3);
		\draw (7331) -- (4);

		\node[circle,inner sep=1.5pt,fill](5) at (3,0) {};
		\node[circle,inner sep=1.5pt,fill](6) at (5,0) {};
		\node[circle,inner sep=1.5pt,fill](7) at (3,1) {};
		\node[circle,inner sep=1.5pt,fill](8) at (5,1) {};
		\node[circle,inner sep=1.5pt,fill](9) at (3,2) {};
		\node[circle,inner sep=1.5pt,fill](10) at (5,2) {};
		\draw (5) -- (6); \draw (5) -- (8); \draw (5) -- (10);
		\draw (7) -- (6); \draw (7) -- (8); \draw (7) -- (10);
		\draw (9) -- (6); \draw (9) -- (8); \draw (9) -- (10);
		
		\node[circle,inner sep=1.5pt,fill](11) at (7,1) {};
		\node[circle,inner sep=1.5pt,fill](12) at (6,1) {};
		\node[circle,inner sep=1.5pt,fill](13) at (6.5,2) {};
		\node[circle,inner sep=1.5pt,fill](14) at (6.5,0) {};
		\node[circle,inner sep=1.5pt,fill](15) at (7.5,0) {};
		\node[circle,inner sep=1.5pt,fill](16) at (7.5,2) {};
		\node[circle,inner sep=1.5pt,fill](17) at (8,1) {};
		\draw (11) -- (12); \draw (11) -- (13); \draw (11) -- (14);
		\draw (11) -- (15); \draw (11) -- (16); \draw (11) -- (17);
		\draw (12) -- (13); \draw (16) -- (17); \draw (14) -- (15);
		
		\node[circle,inner sep=1.5pt,fill](18) at (9,1.2) {};
		\node[circle,inner sep=1.5pt,fill](19) at (9.4,0.8) {};
		\node[circle,inner sep=1.5pt,fill](20) at (9,0) {};
		\node[circle,inner sep=1.5pt,fill](21) at (9.4,2) {};
		\node[circle,inner sep=1.5pt,fill](22) at (9.8,0) {};
		\node[circle,inner sep=1.5pt,fill](23) at (9.8,1.2) {};
		\node[circle,inner sep=1.5pt,fill](24) at (10.2,0) {};
		\node[circle,inner sep=1.5pt,fill](25) at (11,0) {};
		\node[circle,inner sep=1.5pt,fill](26) at (10.6,0.8) {};
		\node[circle,inner sep=1.5pt,fill](27) at (10.6,2) {};
		\node[circle,inner sep=1.5pt,fill](28) at (10.2,1.2) {};
		\node[circle,inner sep=1.5pt,fill](29) at (11,1.2) {};
		\draw (18) -- (21); \draw (18) -- (23); \draw (23) -- (21);
		\draw (19) -- (20); \draw (19) -- (22); \draw (20) -- (22);
		\draw (24) -- (26); \draw (24) -- (25); \draw (25) -- (26);
		\draw (27) -- (28); \draw (27) -- (29); \draw (28) -- (29);
		
		\node[circle,inner sep=1.5pt,fill](30) at (15,1) {};
		\node[circle,inner sep=1.5pt,fill](31) at (16,0) {};
		\node[circle,inner sep=1.5pt,fill](32) at (16,0.5) {};
		\node[circle,inner sep=1.5pt,fill](33) at (16,1) {};
		\node[circle,inner sep=1.5pt,fill](34) at (16,1.5) {};
		\node[circle,inner sep=1.5pt,fill](35) at (16,2) {};
		\node[circle,inner sep=1.5pt,fill](36) at (17,0) {};
		\node[circle,inner sep=1.5pt,fill](37) at (17,0.5) {};
		\node[circle,inner sep=1.5pt,fill](38) at (17,1) {};
		\node[circle,inner sep=1.5pt,fill](39) at (17,1.5) {};
		\node[circle,inner sep=1.5pt,fill](40) at (17,2) {};
		\draw (30) -- (31); \draw (30) -- (32); \draw (30) -- (33);
		\draw (30) -- (34); \draw (30) -- (35); \draw (31) -- (36);
		\draw (32) -- (37); \draw (33) -- (38); \draw (34) -- (39);
		\draw (35) -- (40);
		
		\node[circle,inner sep=1.5pt,fill](41) at (12,0) {};
		\node[circle,inner sep=1.5pt,fill](42) at (13,0) {};
		\node[circle,inner sep=1.5pt,fill](43) at (14,0) {};
		\node[circle,inner sep=1.5pt,fill](44) at (12,0.5) {};
		\node[circle,inner sep=1.5pt,fill](45) at (13,0.5) {};
		\node[circle,inner sep=1.5pt,fill](46) at (14,0.5) {};
		\node[circle,inner sep=1.5pt,fill](47) at (12,1) {};
		\node[circle,inner sep=1.5pt,fill](48) at (13,1) {};
		\node[circle,inner sep=1.5pt,fill](49) at (14,1) {};
		\node[circle,inner sep=1.5pt,fill](50) at (12,1.5) {};
		\node[circle,inner sep=1.5pt,fill](51) at (13,1.5) {};
		\node[circle,inner sep=1.5pt,fill](52) at (14,1.5) {};
		\node[circle,inner sep=1.5pt,fill](53) at (12,2) {};
		\node[circle,inner sep=1.5pt,fill](54) at (13,2) {};
		\node[circle,inner sep=1.5pt,fill](55) at (14,2) {};
		\draw (41) -- (42); \draw (42) -- (43);
		\draw (44) -- (45); \draw (45) -- (46);
		\draw (47) -- (48); \draw (48) -- (49);
		\draw (50) -- (51); \draw (51) -- (52);
		\draw (53) -- (54); \draw (54) -- (55);
    \end{tikzpicture}
  \caption{%
		\label{fig:IndSubs}%
    Example graphs from each of the six minimal graph classes that do not have 
    bounded weak vertex-cover number according to Lemma~\ref{lem:ramsey}:
    $K_6, K_{3,3}, W_3, 4\cdot K_3$, $5\cdot P_2$, and $SS_5$.
	}
\end{figure}
\begin{lemma} 
  \label{lem:ramsey}
We say that a class $\classH$ contains another class $\mathcal{C}$ as induced subgraphs if, 
for\linebreak every~$C \in \mathcal{C}$, there is some $H \in \classH$ such that $H$ contains $C$ as induced subgraph.
If~$\classH$ is a class of graphs with unbounded weak vertex-cover number,
then $\classH$ contains at least one of the following classes as induced subgraphs:
\begin{enumerate}[leftmargin=3em,label=(\roman*)]
    \item the class of all cliques,
    \item the class of all bicliques,
    \item the class of all subdivided stars, 
    \item the class of all windmills, 
    \item the class of all triangle packings, or
    \item the class of all wedge packings.  
  \end{enumerate}
\end{lemma}
\begin{proof}
	Let~$\classH$ be a graph class of unbounded weak vertex-cover number, and let
	$C\in \N$ be a constant such that all cliques, bicliques, subdivided stars,
	windmills, and triangle packings that occur as induced subgraphs in $\classH$
	have size at most~$C$.
	To prove the lemma, we argue that~$\classH$ contains induced~$P_2$-packings
	of unbounded size.
	To simplify the argument, we assume without loss of generality that $\classH$ is closed under taking
	induced subgraphs.
	Let $\classH'\subseteq\classH$ be the class of all graphs $H\in\classH$ that
	do not contain isolated edges.
	Since $\classH$ has unbounded weak vertex-cover number, the vertex-cover
	number of $\classH'$ is unbounded.
  Curticapean and Marx~\cite[Lemma~5.2]{curticapean2014complexity} prove that, 
  in this situation,
	$\classH'$ contains arbitrarily large cliques, induced bicliques, or induced
	matchings.
	By our assumption, the size of every clique and biclique is at most~$C$.
	Thus for every~$k$, there is a graph~$H_k\in\classH'$ such that~$H_k$ contains
	a size-$k$ matching $M_k\subseteq E(H_k)$ as an induced subgraph.
	
	For every~$k$ and every $e\in M_k$, we choose an arbitrary vertex $v_e\in
	V(H_k)\setminus V(M_k)$ that is adjacent in $H_k$ to one or both endpoints
	of~$e$.
	These vertices exist since~$e$ is not an isolated edge in $H_k$ and $M_k$ is
	an induced matching in~$H_k$.
	Let $N_k=\setc{v_e}{e\in M_k}$ and note that $v_e$ and $v_{e'}$ may coincide for distinct $e,e'\in M_k$.
	Let $A_v\subseteq M_k$ be the set of all $e\in M_k$ such that exactly one
	endpoint of $e$ is adjacent to~$v$, and let $B_v\subseteq M_k$ be the set of
	all $e\in M_k$ that have both their endpoints adjacent to~$v$.
	If $A_v\neq\emptyset$, the graph $H_k[V(A_v)\cup \set{v}]$ is an induced
	subdivided star of size $\abs{A_v}$.
	Similarly, if $B_v\neq\emptyset$, the graph $H_k[V(B_v)\cup \set{v}]$ is an
	induced windmill of size~$\abs{B_v}$.
	By our assumption on $\classH$, the sets $A_v$ and $B_v$ have size at most $C$
	for all $v\in N_k$ and all $k$.
	
	Before we argue that arbitrarily large $P_2$-packings exist as induced
	subgraphs, we apply Ramsey's theorem to obtain more structure.
	Since $M_k=\bigcup_{v\in N_k} (A_v \cup B_v)$ holds, the set $N_k$ has size at
	least $k/(2C)$.
	Thus the graph class $\setc{H_k[N_k]}{k\in\N}$ is infinite, and since we
	assumed that every clique in $\classH$ has bounded size, Ramsey's theorem
	guarantees the existence of independent sets $I_k\subseteq N_k$ whose sizes
	are unbounded as $k$ grows.
	
	Finally, we construct a large induced packing of triangles and paths of
	length~$2$ using the following greedy procedure: For each $v\in I_k$ with
	$B_v\neq\emptyset$, we select an arbitrary edge $e\in B_v$ to contribute one
	triangle with~$v$, and we remove $A_v\cup B_v$ from $M_k$.
	Similarly, each $v\in I_k$ with $B_v=\emptyset$ and $A_v\neq\emptyset$
	contributes one copy of~$P_2$ and we delete $A_v\cup B_v$ from $M_k$.
	By definition of~$A_v$ and~$B_v$, the vertex~$v$ is not adjacent to any edge
	in $M_k\setminus(A_v \cup B_v)$; moreover, it is not adjacent to any vertex in
	$I_k\setminus\set{v}$.
	Hence the constructed disjoint union of triangles and paths of length~$2$ is
	indeed an induced subgraph of~$H_k$.
	Since all sets $A_v$ and $B_v$ are of size at most $C$, the number of
	components we constructed is at least $\abs{I_k}/(2C)$, which is unbounded as
	$k$ grows.
	By our assumption on~$\classH$, at most~$C$ of the components are triangles,
	and at least $\abs{I_k}/(2C) - C$ components are copies of~$P_2$.
	We conclude that $\classH$ contains arbitrarily large induced $P_2$-packings.
\end{proof}

Since hereditary classes~$\classH$ are closed under induced subgraphs, 
Lemma~\ref{lem:ramsey} guarantees that any hereditary class $\classH$ with
unbounded weak vertex-cover number contains at least one of the six graph
families defined above as an actual subset of~$\classH$. We need to prove hardness for
each of these six families: 
\begin{lemma}
	\label{lem:spec_cases}
	If $\classH$ is the class of all cliques, the class of all bicliques, the class of all subdivided stars, the
	class of all windmills, the class of all triangle packings, or the class of all wedge packings, then $\#\eiHom(\classH)$ is $\sw$-hard.
\end{lemma}

As we show in the following, every edge-injective homomorphism from a clique, a biclique, or a windmill
into a graph is, in fact, an embedding.
For these three graph families, counting edge-injective homomorphisms is thus equivalent to the
corresponding subgraph counting problem.
Since the families have unbounded vertex-cover number, the main theorem
of Curticapean and Marx~\cite{curticapean2014complexity} implies that the 
subgraph counting
problem for these three graph families is $\sw$-hard.

\begin{lemma}\label{lem:cliques}%
	Let $G$ be a simple graph and let $H$ be a clique, biclique, or windmill.
	Then every edge-injective homomorphism~$\varphi$ from $H$ to $G$ is an embedding.
\end{lemma}
\begin{proof}
	Let $\varphi$ be an edge-injective homomorphism from~$H$ to~$G$.
	For two distinct vertices $x$ and~$y$ of~$H$, we
	have~$\varphi(x)\ne\varphi(y)$ if~$x$ and~$y$ are joined by an edge of~$H$ or
	if they have a common neighbor~$z$ in~$H$.
	If $H$ is a clique, then all $x,y\in V(H)$ with $x\ne y$ are adjacent in~$H$.
	If $H$ is a biclique or a windmill, then any two distinct vertices~$x$ and~$y$
	are either adjacent or have a common neighbor.
	In either case, $\varphi$ is an embedding.
\end{proof}

\begin{proposition}\label{lem:clique hardness}
The problem $\#\eiHom(\classH)$ is $\sw$-hard if $\classH$ is the class of all cliques, the class of all bicliques, or the class of all windmills.
\end{proposition}
\begin{proof}
	By Lemma~\ref{lem:cliques}, the problem $\#\eiHom(\classH)$ is equivalent to $\#\Emb(\classH)$.
  Thus, since~$\classH$ has unbounded vertex-cover number, the problem is $\sw$-hard 
  by the dichotomy for counting embeddings~\cite[Theorem~1.1]{curticapean2014complexity}.
\end{proof}

For the class of triangle packings, we devise a straightforward reduction from 
the problem of counting $k$-matchings in bipartite graphs.
Essentially, we add an additional vertex that is adjacent to all other vertices, 
and since the original graph was bipartite, every triangle of the triangle 
packing must use the new vertex.
\begin{proposition}\label{lem:triangleswindmills}
The problem $\#\eiHom(\classH)$	is $\#\W$-hard if $\classH$ is the class of all triangle packings.
\end{proposition}
\begin{proof}
	We reduce from the problem of counting $k$-matchings in a bipartite graph.
	Given a simple bipartite graph $G=(U \cup V, E)$ and a number~$k$, we construct a graph~$G'$
	from~$G$ by adding a single apex~$a$, that is, a new vertex $a$ and the edges~$\{a,v\}$ for all $v \in U \cup V$.
	Since~$G$ is bipartite, every triangle in~$G'$ consists of $a$ and some
	vertices $u \in U$ and $v \in V$.
	We denote such a triangle by $a_{v,u}$.
	The output of the reduction is the instance $(H,G')$, where~$H$ is the graph
	$k\cdot K_3$.
	In either case, the edges of~$H$ partition into $k$ triangles; let us fix this
	partition and an arbitrary order on the triangles.
	
	Since $G'$ is a simple graph, the homomorphic image of a triangle is a
	triangle, and exactly one vertex of each of the $k$ triangles is mapped to
	$a$.
	Let $\varphi$ be an edge-injective homomorphism from~$k \cdot K_3$ to~$G'$.
	Let the image of the $i$-th triangle be $\set{a,u_i,v_i}$.
	Since~$\varphi$ is edge-injective and $a$ is an apex, all $u_i$ and $v_i$ are
	mutually distinct.
	Moreover, the edges $\set{u_i,v_i}$ form a matching~$M_\varphi$ of size~$k$
	in~$G$.
	
	Finally, we claim that the number~$m_k$ of $k$-matchings of $G$ can be derived
	from the number of edge-injective homomorphisms.
	The homomorphism $\varphi$ can first arbitrarily choose one vertex of each
	triangle to be mapped to~$a$, which gives it $3^k$ choices.
	For the remaining matching of size~$k$, the homomorphism~$\varphi$ must map it
	to a matching in~$G$ by edge-injectivity.
	Thus it can choose one of the~$2^k \cdot k!$ automorphisms of the
	$k$-matching.
	Overall, we get $6^k\cdot k!$ edge-injective homomorphisms~$\varphi$ with $M_\varphi=M$.
	Thus $\abs{\eiHom(k\cdot K_3,G')}/ \paren{6^k \cdot k!}$.
	The reduction takes polynomial time, increases the parameter from $k$ to
	$\abs{H}=O(k)$, and requires only one query to the oracle.
\end{proof}

For subdivided stars, we reduce from counting $k$-matchings in well-structured bipartite 
graphs to (essentially) the graph $G^0$ that was constructed in the proof of 
Theorem~\ref{thm: wedge packing}.
The analysis is much simpler for subdivided stars since we can guarantee easily 
that the center vertex is mapped to the newly added vertex~$0$.

\begin{proposition}\label{lem:subdividedstars}
The problem $\#\eiHom(\classH)$ is $\#\W$-hard if $\classH$ is the class of all subdivided stars.
\end{proposition}
\begin{proof}
	We reduce from the problem of counting $k$-matchings in bipartite
	graphs where the degree of right-side vertices is at most two and any two
	distinct left-side vertices have at most one common neighbor.
	Let $(G,k)$ be an instance of this problem, and let $L(G)$ and $R(G)$ be the
	left and right vertex sets, respectively.
	Starting from $G$, we construct a new graph $G'$ (see Figure~\ref{fig:gprime}):
	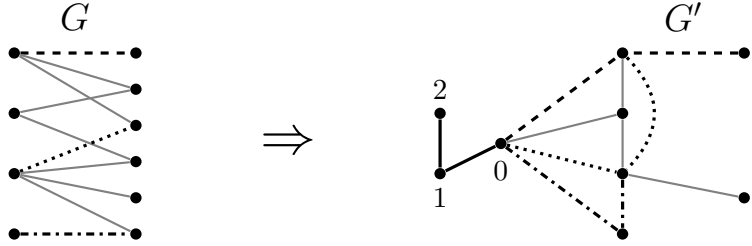
\begin{figure}
	\centering
	\input{figures/graphic_gzero.tex}
	\caption{\label{fig:gprime} The construction of $G'$ in the proof of Proposition~\ref{lem:subdividedstars}, including the image of a homomorphism from the subdivided star $SS_4$ such that vertex $2$ is contained in the image: One ray of the subdivided star is mapped to the edges
$\{0,1\},\{1,2\}$ and functions as an anchor. The other three rays (dashed, dotted, and dash-dotted) correspond to a $3$-matching in $G$.}
	\end{figure}
	\begin{enumerate}
		\item
		Insert a vertex $0$ that is adjacent to all vertices of $L(G)$.
		\item
		For every vertex $v\in R(G)$ with $\deg(v)=2$, remove $v$ from the graph
		and add the set $N(v)$ as an edge to $G'$.
		\item
		Add two special vertices $1$ and $2$, as well as the edges $01$ and $12$.
	\end{enumerate}
	Since $G$ is a simple graph and any two distinct vertices $u,v\in L(G)$ have
	at most one common neighbor in~$G$, the graph $G'$ is again a simple graph.
	
	Let $m_k$ be the number of $k$-matchings in $G$, let $H$ be the subdivided star
	of size $k+1$, and let $s(G)=\abs{\eiHom(H,G)}$.
	We claim that 
  \[
  (k+1)!\cdot m_k=s(G') - s(G'-\set{2}).
  \]
	Clearly $s(G')-s(G'-\set{2})$ is exactly the number of edge-injective
	homomorphisms~$\varphi$ from $H$ to~$G'$ such that $2$ is in the image
	of~$\varphi$.
	The claim is that there is a correspondence between such homomorphisms and the
	$k$-matchings in~$G$.
	
	Let $\varphi$ be an edge-injective homomorphism with~$2=\varphi(z)$ 
  for some~${z\in V(H)}$.
	Then $z$ must be a degree-$1$ vertex in~$H$ since~$2$ has exactly one neighbor
	in~$G'$ and $H$ does not contain isolated vertices.
	Let~$y$ be the neighbor of~$z$ and let~$x$ be the center vertex of the
	subdivided star.
	Then~$\varphi(y)=1$ and $\varphi(x)=0$ holds.
	Next, let $y_1,\dots,y_k$ be the other degree-$2$ vertices of $H$, and let
	$z_1,\dots,z_k$ be the corresponding degree-$1$ vertices.
	Since $\varphi$ is an edge-injective homomorphism, the vertices
	$a_i:=\varphi(y_i)$ are mutually distinct and satisfy $a_i\in L(G)$.
	
  Note that ${\varphi(z_i)\not\in\set{0,1,2}}$.
	We define the matching $M_\varphi=\set{a_1b_1,\dots,a_kb_k}$ as follows.
	If~$\varphi(z_i)\in R(G)$, then it is a degree-$1$ vertex of $R(G)$, and we
	set $b_i=\varphi(z_i)$.
	Otherwise we have $\varphi(z_i)\in L(G)$ and the edge $\varphi(y_i z_i)$
	exists in~$G'$; we let $b_i$ be the unique vertex with~$N_G(b_i)=\set{a_i,\varphi(z_i)}$ that caused this edge to be added to~$G'$ in
	the construction.
	The $b_i$ are mutually distinct due to the edge-injectivity.
	Hence~$M_\varphi$ is indeed a matching.
	
	For every $k$-matching~$M$ of $G$, there are $(k+1)!$ distinct edge-injective homomorphisms $\varphi$ with
	$M=M_\varphi$ since~$\varphi$ can choose an arbitrary order for the $k+1$ rays
	of the subdivided star.
	This proves the claim.

	Overall, the reduction runs in polynomial time and queries the oracle exactly
	two times with parameter $\abs{H}=O(k)$.
\end{proof}

Now we have established hardness for all of the minimal cases:

\begin{proof}[Proof of Lemma~\ref{lem:spec_cases}]
The $\sw$-hardness of counting wedge packings follows from Theorem~\ref{thm: wedge packing}. The remaining five cases were treated in this section.
\end{proof}

As a consequence, we obtain Theorem~\ref{thm:dichotomy}.
\begin{proof}[Proof of Theorem~\ref{thm:dichotomy}]
For classes with bounded weak vertex-cover number, an algorithm with polynomial running time is given in Theorem~\ref{thm:polytime}. For every hereditary class of unbounded weak vertex-cover number, Lemma~\ref{lem:ramsey} and Lemma~\ref{lem:spec_cases} together give $\sw$-hardness.
\end{proof}

\subsection{Hardness for cycles and paths}

The dichotomy theorem for $\#\eiHom(\ch)$ with hereditary graph classes $\classH$ leaves open some non-hereditary graph classes of interest. 
In this final part of the paper, we investigate $\#\eiHom(\ch)$ for the class of cycles and that of paths
and prove $\sw$-hardness for these problems.

\begin{theorem} 
\label{thm: edge-disj-path}
	For the classes $\mathcal{C}$ and $\mathcal{P}$ of all cycles and paths, 
  respectively, the problems $\#\eiHom(\mathcal{C})$ and 
  $\#\eiHom(\mathcal{P})$ are $\sw$-hard. 
\end{theorem}

We point out that the problems of counting edge-injective homomorphisms from $C_k$ and~$P_k$ are equivalent to the problems of counting edge-disjoint $k$-cycles and edge-disjoint $k$-paths, respectively. In particular, for any graph $G$, we have that
$\#\eiHom(C_k, G)$ equals $2k$ times the number of edge-disjoint $k$-cycles in $G$, 
while $\#\eiHom(P_k,G)$ equals twice the number of edge-disjoint $k$-paths in $G$.

We will first show that $\#\eiHom(\mathcal{C})$ is $\#\W$-hard. To this end, we consider the edge-weighted version of counting edge-injective homomorphisms in an intermediate step. Let $H$ and $G$ be graphs and let $w:E(G) \rightarrow \mathbb{N}$ a weight-function. The number of edge-weighted edge-injective homomorphisms is defined as follows
\begin{equation}
\#\eiHom(H,G,w) := \sum_{\varphi \in \eiHom(H,G)} ~\prod_{e \in E(H)} w(\varphi(e)) \,.
\end{equation}
Then the problem $\#\mathsf{W}\eiHom(\mathcal{H})$ asks, given a graph $H \in \mathcal{H}$ and an arbitrary graph~$G$ with weight-function~$w$, to compute this quantity. The parameter is $|V(H)|+ \max\{w(e)~|~e \in E(G)\}$. That is, the edge-weights of $G$ must be bounded by some function in the size of the pattern graph $H$.
\begin{lemma}
	\label{lem:wecyc}
$\#\mathsf{W}\eiHom(\mathcal{C})$ is $\sw$-hard.
\end{lemma}
\begin{proof}
First we observe that, for all $k \in \mathbb{N}$, we have
\begin{equation}
\label{eq: weighted-einj-cycle}
\#\eiHom(C_k,G,w) = 2k \cdot \sum_{c \in \EC_k(G)}\prod_{e \in c}w(e)
\end{equation}
where $\EC_k(G)$ denotes the set of all edge-disjoint cycles of length $k$ in $G$.
We show $\#\W$-hardness by constructing an fpt Turing reduction from the $\sw$-hard problem $\#\Sub(\mathcal{C})$ of counting simple cycles of length $k$, see \cite{FlumGrohe,curticapean2014complexity} for hardness proofs of this problem.
On input a graph $G$ and $k \in \N$, our reduction proceeds as follows: If $k < 3$, then return $0$. 
Else consider the graph $G_x$ obtained from~$G$ by substituting each node $v\in V$, 
of some degree $d_v$, 
by the gadget graph $H_v$ constructed as follows: 
We start with a path of length $3$ whose intermediate edge has weight $x$.\footnote{Here, $x$ is an indeterminate, so the quantity \eqref{eq: weighted-einj-cycle} is a polynomial in $x$.} 
Next we add vertices $s^1_v,\ldots,s^{d_v}_v$ and connect each of them with an edge to one endpoint of the path. After that we add vertices $t^1_v,\ldots,t^{d_v}_v$ and connect each of them with an edge to the other endpoint.

\begin{figure}[tp]
  \centering
\begin{tikzpicture}[-,thick,scale=0.6,main node/.style={circle,inner sep=1.5pt,fill}, rotate =-90]
  \node[main node] (0) at (1.5,3) {};
  \node[main node,label={$s^1_v$}] (1) at (-0.5,4.5) {};
  \node[main node,label={$s^2_v$}] (2) at (1.5,4.5) {};
  \node[main node,label=below:{$s^{d_v}_v$}] (3) at (3.5,4.5) {};
  \node (4) at (2.25,4.5) {$\vdots$};
  \node[main node,label={$t^1_v$}] (5) at (-0.5,-2.5) {};
  \node[main node,label={$t^2_v$}] (6) at (1.5,-2.5) {};
  \node[main node,label=below:{$t^{d_v}_v$}] (7) at (3.5,-2.5) {};
  \node[main node] (10) at (1.5,-0.5) {};
  \node[main node] (11) at (1.5,0.5) {};
  \node[main node] (12) at (1.5,2) {};
  \node (8) at (2.25,-2.5) {$\vdots$};
  \draw (0) -- (1) ;
  \draw (0) -- (2);
  \draw (0) -- (3);
  \draw (5) -- (10);
  \draw (6) -- (10);
  \draw (7) -- (10);
  \draw (11) -- (10);
  \path (12) edge node [above] {$x$} (11);
  \draw (12) -- (0);

\end{tikzpicture}
  \caption{%
  \label{fig:HvGi}%
  Gadget $H_v$ for a vertex $v$ of degree $d_v$ as used in the proof of Lemma~\ref{lem:wecyc}.
  }
\end{figure}
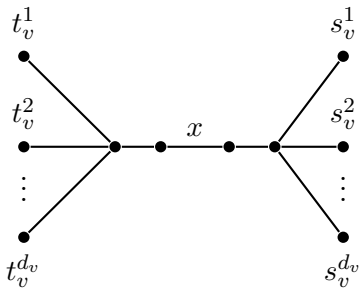
Furthermore add edges $\{s^i_v,t^j_u\}$ and $\{ s^j_u,t^i_v \}$ for every edge $\{u,v\}$ that is the $i$-th edge of $v$ and the $j$-th edge of $u$. The resulting graph is shown in Figure~\ref{fig:HvGi}. Consider $G_x$ as a weighted graph were every edge has weight $1$ except for the edges labeled with $x$ as above. Now, querying the oracle for $\#\mathsf{W}\eiHom(\mathcal{C})$ with input~$C_{6k}$ and $G_x$, and dividing by $12k$ yields a polynomial $p \in \Z[x]$.
\begin{claim}
The degree of $p$ is bounded by $k$. Furthermore the coefficient of $x^k$ equals twice the number of simple $k$-cycles in $G$.
\end{claim}
\begin{claimproof}
  The shortest edge-disjoint path between any pair of two different edges with
  weight $x$ is at least $5$ (excluding the two edges with weight $x$ from the length.) As
  we search for edge-disjoint cycles of length $6k$, the weight $x$ can occur at most
  $\frac{6k}{5+1} = k$ times in one cycle. Therefore the degree of $p$ is
  bounded by $k$. Furthermore the distance is equal to $5$ if and only if the
  two edges belong to gadgets~$H_v$ and~$H_u$ such that $\{v,u\} \in E(G)$. In
  particular, this path either leaves $H_v$ through $s^i_v$ and enters $H_u$
  through $t^j_u$ for some $i$ and $j$ or it leaves $H_v$ through $t^{i'}_v$ and
  enters~$H_u$ through $s^{j'}_u$ for some $i'$ and $j'$. 

  Now consider an
  edge-disjoint cycle $c$ of length $6k$ that includes $k$ edges with weight
  $x$. It follows that $c=(e_1,P_1,\ldots,e_k,P_k,e_1)$ where each $e_i$ has
  weight $x$ and each $P_i$ is a path consisting of $5$ edges with weight $1$.

  Next let $H_{v_i}$ be the gadget containing $e_i$ and consider $H_{v_1}$. It
  holds that $P_1$ either passes through $s_{v_1}^i$ and $t_{v_2}^j$ for some
  $i$ and $j$ or through $t_{v_1}^{i'}$ and $s_{v_2}^{j'}$ for some $i'$ and
  $j'$. However, if we fix one of these two options, only one possibility remains
  for all other $P_2,\ldots,P_k$ as we cannot turn around in a gadget if we
  consider edge-disjoint cycles. Therefore there are exactly two edge-disjoint
  cycles $c_1=(e_1,P_1,\ldots,e_k,P_k,e_1)$ and
  $c_2=(e_1,P'_1,\ldots,e_k,P'_k,e_1)$ that correspond to the cycle
  $c=(v_1,\ldots,v_k,v_1)$ in $G$ and vice versa. Furthermore $c$ is simple as
  $c_1,c_2$ are edge-disjoint, that is, the $e_i$'s and therefore the $v_i$'s are pairwise different.
\end{claimproof}
To conclude the proof of Lemma~\ref{lem:wecyc}, we compute the coefficient of $x^k$ in the degree-$k$ polynomial $p$ by means of polynomial interpolation from the evaluations $p(0), \ldots , p(k)$. These evaluations are obtained by oracle calls to $\#\mathsf{W}\eiHom(\mathcal{C})$ with input $C_{6k}$ and $G_b$ for $b = 0,\ldots,k$ (and dividing by $12 k$). As the edge-weights of every graph $G_b$ are bounded by $k$, the overall parameter $|V(C_{6k})| + \max\{w(e) ~|~e \in G_b\}$ is bounded by $7k$, proving that this reduction is indeed an fpt Turing-reduction.
\end{proof}
We show hardness of the unweighted version by reduction from the weighted version;
this requires us to devise a strategy for removing weights.
\begin{lemma}
	\label{lem:ecyc}
There is an fpt Turing reduction from $\#\mathsf{W}\eiHom(\mathcal{C})$ to $\#\eiHom(\mathcal{C})$.
\end{lemma}
\begin{proof}
The input for the reduction is a number $k \in \N$ and an edge-weighted graph~$G$ whose edge weights are bounded by $k$. We assume $k \geq 4$, as we can otherwise solve the problem in polynomial time by brute-force. The following gadgets will be used in the reduction:
\begin{itemize}
\item $G_1$ is simply one undirected edge $e_1:=\{a_1,b_1\}$
\item $G_{i+1}$ is constructed from $G_i$ as follows: We add vertices $a_{i+1}$ and $b_{i+1}$ and edges $\{a_{i+1},a_i\}$ and $\{b_{i+1},b_i\}$. Furthermore we add a path of length $2i+1$ between $a_{i+1}$ and $b_{i+1}$ and denote the $i+1$-th edge of this path $e_{i+1}$. $G_{i+1}$ is depicted in Figure~\ref{fig:HvGi2}.
\end{itemize}
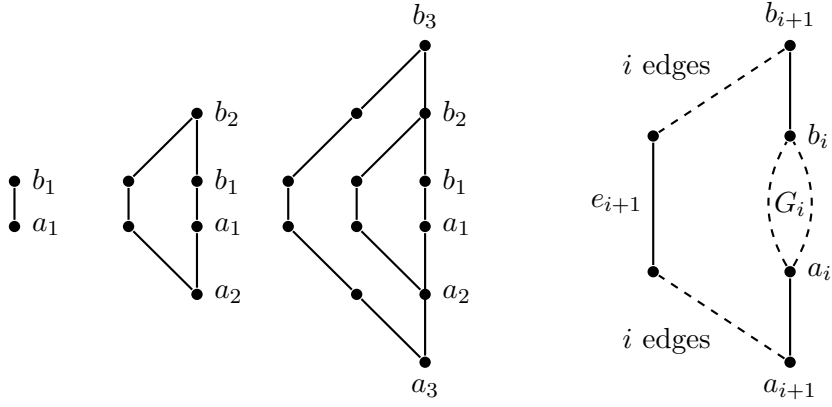
\begin{figure}[tp]
  \centering
\begin{tikzpicture}[-,thick,scale=0.6,main node/.style={circle,inner sep=1.5pt,fill}]

   \node[main node,label=right:{$a_1$}] (0110) at (-7,0.5) {};
   \node[main node,label=right:{$b_1$}] (0112) at (-7,1.5) {};
   \draw (0110) -- (0112);
    
   \node[main node,label=right:{$a_1$}] (0110) at (-3,0.5) {};
   \node[main node,label=right:{$b_1$}] (0112) at (-3,1.5) {};
   \node[main node] (0119) at (-4.5,0.5) {};
   \node[main node] (0120) at (-4.5,1.5) {};
   \node[main node,label=right:{$a_2$}] (0121) at (-3,-1) {};
   \node[main node,label=right:{$b_2$}] (0122) at (-3,3) {};

   \draw (0110) -- (0112);
   \draw (0121) -- (0119);
   \draw (0122) -- (0120);
   \draw (0119) -- (0120);    
   \draw (0110) -- (0121);
   \draw (0112) -- (0122);

  %\path (0) edge[style = transparent] node [right] {test} (1);
   \node[main node,label=right:{$a_1$}] (0110) at (2,0.5) {};
   \node[main node,label=right:{$b_1$}] (0112) at (2,1.5) {};
   \node[main node,label=above:{$b_{3}$}] (0113) at (2,4.5) {};
   \node[main node,label=below:{$a_{3}$}] (0114) at (2,-2.5) {};
   \node[main node] (0115) at (-1,0.5) {};
   \node[main node] (0116) at (-1,1.5) {};
   \node[main node] (0117) at (0.5,-1) {};
   \node[main node] (0118) at (0.5,3) {};
   \node[main node] (0119) at (0.5,0.5) {};
   \node[main node] (0120) at (0.5,1.5) {};
   \node[main node,label=right:{$a_2$}] (0121) at (2,-1) {};
   \node[main node,label=right:{$b_2$}] (0122) at (2,3) {};
   
   \path (0114) edge node {} (0110);
   \path (0112) edge node {} (0113);
   \draw (0110) -- (0112);
   \path (0114) edge node {} (0117);
   \path (0113) edge node {} (0118);
   \path (0115) edge node {} (0116);
   \draw (0118) -- (0116);
   \draw (0117) -- (0115);
   \draw (0121) -- (0119);
   \draw (0122) -- (0120);
   \draw (0119) -- (0120);

  %\path (0) edge[style = transparent] node [right] {test} (1);
   \node[main node,label=right:{$a_i$}] (110) at (10,-0.5) {};
   \node (111) at (10,1) {$G_i$};
   \node[main node,label=right:{$b_i$}] (112) at (10,2.5) {};
   \node[main node,label=above:{$b_{i+1}$}] (113) at (10,4.5) {};
   \node[main node,label=below:{$a_{i+1}$}] (114) at (10,-2.5) {};
   \node[main node] (115) at (7,-0.5) {};
   \node[main node] (116) at (7,2.5) {};
   
   \path (114) edge node {} (110);
   \path (112) edge node {} (113);
   \path[bend right, dashed] (110) edge (112);
   \path[bend left, dashed] (110) edge (112);
   \path[dashed] (114) edge node [below left] {$i$ edges} (115);
   \path[dashed] (113) edge node [above left] {$i$ edges}(116);
   \path (115) edge node [left] {$e_{i+1}$} (116);

\end{tikzpicture}
  \caption{%
  \label{fig:HvGi2}%
  Graphs $G_1$, $G_2$ and $G_3$, as well as the inductive construction of graph $G_{i+1}$ as used in the proof of Lemma~\ref{lem:ecyc}.
  }
\end{figure}
It is easy to see that $|V(G_k)| \leq O(k^2)$.
\begin{claim}
  \label{cl:cyc2}
  For every $k \geq 1$, there are exactly $k$ edge-disjoint walks from $a_k$ to $b_k$ in~$G_k$, each of length $2k-1$. Furthermore $e_j$ is contained in exactly one of this walks for every $j \in [k]$.
\end{claim}
\begin{claimproof}
  We prove the claim by induction on $k$; it is obvious for $k=1$. For the induction step, consider $G_{k+1}$: An edge-disjoint walk from $a_{k+1}$ to $b_{k+1}$ either takes the "left" way and therefore contains $e_{k+1}$ or takes a way through $G_k$:
    \begin{itemize}
      \item The "left" way has length $2k + 1 = 2(k+1)-1$.
      \item Every way through $G_k$ corresponds one-to-one to a closed walk from $a_k$ to $b_k$ in $G_k$. Applying the induction hypothesis we obtain that there are exactly $k$ edge-disjoint walks from $a_k$ to $b_k$ in $G_k$, one for every $e_j$ for $j \in [k]$. Furthermore each of this walks has length~$2k-1$. It follows that there are exactly $k$ edge-disjoint walks from $a_{k+1}$ to $b_{k+1}$, each of length $2 + 2k - 1 = 2(k+1)-1$. $e_j$ is contained in exactly one of this walks for every $j \in [k]$.
    \end{itemize}
    We conclude that the claim is fulfilled for $G_{k+1}$.
\end{claimproof}
It follows that the longest edge-disjoint cycle in $G_k$ has length $2\cdot (2k + 1) = 4k + 2$.
Let $W$ be the maximum weight of an edge, where $W \leq k$. Now let $H_i$ be the gadget constructed from~$G_{W}$ by removing edges $e_{W},\cdots, e_{i+1}$. We have $H_{W} = G_{W}$. Applying Claim~\ref{cl:cyc2}, we obtain that there are exactly $i$ edge-disjoint walks from~$a_{W}$ to $b_{W}$ in $H_i$. Furthermore each of this walks has length $2W-1$. Finally, we construct~$G'$ from $G$ by substituting each edge $e = \{a,b\}$ with~$H_{w(e)}$ and edges~$\{a,a_{W}\}$ and $\{b,b_{W}\}$.
\begin{claim}
\label{cl:cyc3}
The number of edge-disjoint cycles of length $(2Wk+k)$ in $G'$ equals \[\sum_{c \in \EC_k(G)}\prod_{e \in c}w(e) \,. \]
\end{claim}
\begin{claimproof}
Consider an edge-disjoint cycle $c$ of length $(2Wk+k)$ in $G'$.
Assuming $c$ does contain an edge $\{a,a_{W} \}$ (that is, it is not entirely contained in one gadget), it follows that $c$ can cross every $a_{W}$ and $b_{W}$ at most once. To see this, observe that every time when such a node is reached we can consider the cycle coming from "outside the gadget" (e.g. by choosing a fitting orientation of $c$). Since $c$ is an edge-disjoint cycle, we have to continue by an edge-disjoint walk through the end of the gadget. This walk has length $2W-1$ by Claim~\ref{cl:cyc2}. Now we cannot turn around inside the gadget again and complete the cycle afterwards since otherwise we would have constructed a longer edge-disjoint walk from one endpoint of the gadget to the other, which contradicts Claim~\ref{cl:cyc2}. It follows that every edge-disjoint cycle of length $(2Wk+k)$ that is not entirely contained in one gadget consists of $2W-1$ walks through gadgets. Now, taking an edge $e = \{a,b\}$ with weight $w(e)$ in~$G$ corresponds to taking one of the $w(e)$ edge-disjoint walks $(a,a_{W},\cdots,b_{W},b)$ of length $2W-1+2$ through $H_{w(e)}$ in $G'$. As $k\cdot (2W-1+2) = (2Wk+k)$ it follows that an edge-disjoint cycle of length $k$ in $G$ corresponds to the edge-disjoint cycles of length $(2Wk+k)$ in $G'$ that cross the gadgets corresponding to the weighted edges in $G$, but only if no edge-disjoint cycle of length $(2Wk+k)$ entirely fits in one gadget. However, the latter cannot be the case since the longest edge-disjoint cycle in $G_{W}$ has length $4W+2$ and for every~$k>2$ it holds that
\[4W+2 < 4W \leq 2Wk < 2Wk + k \,. \]  
\end{claimproof}
Now, using Claim~\ref{cl:cyc3}, Equation~\ref{eq: weighted-einj-cycle} and the fact that for every graph $G$ and $k \in \mathbb{N}$, it holds that $\#\eiHom(C_k \to G)$ equals $2k$ times the number of edge-disjoint $k$-cycles in $G$, we obtain that
\begin{align*}
\#\eiHom(C_{2Wk+k},G') &= 2(2Wk+k) \sum_{c \in \EC_k(G)}\prod_{e \in c}w(e) \\~&= \frac{2(2Wk+k)}{2k}\cdot \#\eiHom(C_k,G,w)
\\~&= (2W+1)\cdot \#\eiHom(C_k,G,w)\,.
\end{align*}
The above reduction is indeed an fpt Turing reduction, as $G'$ can be constructed in time $O(n^2\cdot k^2)$ and the value of the new parameter is $2Wk+k \leq O(k^2)$. This concludes the proof.
\end{proof}

\begin{corollary}
\label{cor:cycles_hard}
  The problem $\#\eiHom(\mathcal{C})$ is $\sw$-hard.
\end{corollary}

It remains to show hardness for $\mathcal{P}$:
\begin{lemma}
\label{lem:paths_hard}
The problem $\#\eiHom(\mathcal{P})$ is $\sw$-hard.
\end{lemma}
\begin{proof}
We will reduce from $\#\eiHom(\mathcal{C})$. First, we let $\EC_k(G,v)$ be the set of all edge-disjoint cycles of length $k$ in $G$ that contain $v\in V(G)$. Recall, for comparison, that $\EC_k(G)$ denotes the set of all edge-disjoint cycles of length $k$ in $G$.
\begin{claim}
    \label{cl:cycind}
     It holds that
    \[\EC_k(G) = \bigcup_{i=1}^{|V(G)|} \EC_k(G-\{v_{i+1},\ldots,v_n\},v_i),\]
    where the union is in fact a disjoint union.
  \end{claim}
  \begin{claimproof}
    By induction on $|V(G)|$. If $|V(G)|=0$, the union is empty and therefore the claim holds. Otherwise let $|V(G)|=n+1$. It holds that
    \begin{equation*}
    \begin{split}
    \EC_k(G) &= \EC_k(G,v_{n+1}) \dotcup (\EC_k(G)\setminus \EC_k(G,v_{n+1}))\\
    & = \EC_k(G,v_{n+1}) \dotcup \EC_k(G-\{v_{n+1}\})\\
    & = \EC_k(G,v_{n+1}) \dotcup (\dotbigcup_{i=1}^{n} \EC_k( (G-\{v_{n+1}\})-\{v_{i+1},\cdots,v_{n} \},v_i ))\\
    & = \EC_k(G,v_{n+1}) \dotcup (\dotbigcup_{i=1}^n \EC_k(G-\{ v_{i+1},\cdots,v_{n+1} \},v_i))\\
    & = \dotbigcup_{i=1}^{n+1} \EC_k(G-\{v_{i+1},\ldots,v_{n+1}\},v_i)
    \end{split}
    \end{equation*}
    Here, the third equality follows from the induction hypothesis.
  \end{claimproof}
   It follows that
   \begin{equation}
   \label{eq:eck_eckv}
  |\EC_k(G)| = \sum_{i=1}^{|V(G)|} |\EC_k(G-\{v_{i+1},\ldots,v_n\},v_i)|\,.    
   \end{equation}
Now let $G_i = G-\{v_{i+1},\ldots,v_n\}$. We show that $|\EC_k(G_i,v_i)|$ can be computed using an oracle for $\#\eiHom(\mathcal{P})$:

First, we construct the graph $G'_i$ by adding vertices $s$ and $t$ and edges $\{s,v_i\}$ and~$\{t,v_i\}$. For~$M\subseteq \{s,t\}$ let $A_{i,M}$ be the set of edge-disjoint paths of length~$k+2$ that do not pass through a vertex $u \in M$ and let $G'_{i,M}$ be the graph obtained from~$G'_i$ by removing every vertex that lives in $M$. Note that $|A_M|$ can be computed by querying the oracle for $P_{k+1}$ and $G'_{i,M}$ (and dividing by 2). Now it holds that for all $i\in \{1,\dots,|V(G)|\}$: 
\begin{align*}
|\EC_k(G_i,v_i)| = |A_{i,\emptyset} \setminus (A_{i,\{s\}}\cup A_{i,\{t\}})| = |A_{i,\emptyset}|-|A_{i,\{s\}}|-|A_{i,\{t\}}| + |A_{i,\{s,t\}}| \,,
\end{align*}
where the last equality follows from the principle of inclusion and exclusion. Finally the values of $|\EC_k(G_i,v_i)|$ for all $i \in \{1,\dots,|V(G)|\}$ allow us to compute $|\EC_k(G)|$ (see Equation~\ref{eq:eck_eckv}), which equals $1/2 \cdot \#\eiHom(C_k,G)$.
\end{proof}

\begin{proof}[Proof of Theorem~\ref{thm: edge-disj-path}]
  Follows from Corollary~\ref{cor:cycles_hard} and Lemma~\ref{cor:cycles_hard}.
\end{proof}

%%%%%%%%%%%%%%%%%%%%%%%%%%%%%%%%%%%%%%%%%%%%%%%%%%%%%%%%%%%%%%%%%%%%%%%%%%%%%%
\paragraph*{Acknowledgments.}~%
The authors thank Cornelius Brand and Markus Bl\"aser for interesting discussions, 
and Johannes Schmitt for pointing out a proof of Lemma~\ref{lem:interpolation} and allowing us to use it in this paper.

\bibliographystyle{plainurl}
\bibliography{references}{}
\end{document}

%% file: figures/graphic_gr.tex
\begin{tikzpicture}[-,thick, xscale = 0.4, yscale=0.4]
  \node[circle,inner sep=1.5pt,fill](1) at (-2,1) {};
  \node[circle,inner sep=1.5pt,fill](2) at (-2,-0.666) {};
  \node[circle,inner sep=1.5pt,fill](3) at (-2,-2.333) {};
  \node[circle,inner sep=1.5pt,fill](4) at (-2,-4) {};
  \node[circle,inner sep=1.5pt,fill](5) at (0,1) {};
  \node[circle,inner sep=1.5pt,fill](6) at (0,0) {};
  \node[circle,inner sep=1.5pt,fill](7) at (0,-1) {};
  \node[circle,inner sep=1.5pt,fill](8) at (0,-2) {};
  \node[circle,inner sep=1.5pt,fill](9) at (0,-3) {};
  \node[circle,inner sep=1.5pt,fill](10) at (0,-4) {};
  
  \draw (1)--(5);
  \draw (1)--(6);
  \draw (1)--(7);
  \draw (2)--(6);
  \draw (2)--(8);
  \draw (3)--(7);
  \draw (3)--(8);
  \draw (3)--(9);
  \draw (3)--(10);
  \draw (4) -- (10);
  
  \node[circle,inner sep=1.5pt,fill](11) at (8,1) {};
  \node[circle,inner sep=1.5pt,fill](12) at (8,-0.666) {};
  \node[circle,inner sep=1.5pt,fill](13) at (8,-2.333) {};
  \node[circle,inner sep=1.5pt,fill](14) at (8,-4) {};
  \node[circle,inner sep=1.5pt,fill](15) at (10,1) {};
  \node[circle,inner sep=1.5pt,fill](16) at (10,-3) {};
  
  \node[circle,inner sep=1.5pt,fill](17) at (6,-1.5) {};
  
  \node[circle,inner sep=1.5pt,fill](18) at (5,-0.666) {};
  \node[circle,inner sep=1.5pt,fill](19) at (5,-2.333) {};
%  \node (20) at (5.25, -1.325) {{\tiny \bfseries $\vdots$}};
  \draw[decoration={brace,mirror,raise=5pt},decorate]
  (5,-0.5) -- node[left=6pt] {$r$} (5,-2.5);
  
  \draw (18) -- (17);
  \draw (19) --(17);
  \draw (17) --(12);
  \draw (17) -- (11);
  \draw (15) -- (11);
  \draw (17) -- (13);
  \draw (11) to [bend left] (13);
  \draw (17) -- (14);
  \draw (13) -- (14);
  \draw (12) -- (13);
  \draw (16) -- (13);
  \draw (11) --(12);
  
  \node (21) at (2.5, -1.5) {{\huge $\Rightarrow$}};
%  \node (22) at (-1, 2) {{\Large $G$}};
%  \node (23) at (9, 2) {{\Large $G^r$}};
  \node (24) at (4,-5) {\footnotesize Construction of $G^r$};

\end{tikzpicture}

%% file: figures/improved_graphic.tex
\begin{tikzpicture}[-, xscale = 0.45, yscale=0.35]
\node (1337) at (0,3) {$~$};
  \node[circle,inner sep=1.5pt,fill](1) at (-2,1) {};
  \node[circle,inner sep=1.5pt,fill](2) at (-2,-0.666) {};
  \node[circle,inner sep=1.5pt,fill](3) at (-2,-2.333) {};
  \node[circle,inner sep=1.5pt,fill](4) at (-2,-4) {};
  \node[circle,inner sep=1.5pt,fill](5) at (0,1) {};
  \node[circle,inner sep=1.5pt,fill](6) at (0,0) {};
  \node[circle,inner sep=1.5pt,fill](7) at (0,-1) {};
  \node[circle,inner sep=1.5pt,fill](8) at (0,-2) {};
  \node[circle,inner sep=1.5pt,fill](9) at (0,-3) {};
  \node[circle,inner sep=1.5pt,fill](10) at (0,-4) {};
  
  \draw[very thick, dashed] (1)--(5);
  \draw[thick, gray] (1)--(6);
  \draw[thick, gray] (1)--(7);
  \draw[thick, gray] (2)--(6);
  \draw[thick, gray] (2)--(8);
  \draw[very thick, dotted] (3)--(7);
  \draw[thick, gray] (3)--(8);
  \draw[thick, gray] (3)--(9);
  \draw[thick, gray] (3)--(10);
  \draw[very thick, dashdotted] (4) -- (10);
  
  \node[circle,inner sep=1.5pt,fill](11) at (6,1) {};
  \node[circle,inner sep=1.5pt,fill](12) at (6,-0.666) {};
  \node[circle,inner sep=1.5pt,fill](13) at (6,-2.333) {};
  \node[circle,inner sep=1.5pt,fill](14) at (6,-4) {};
  \node[circle,inner sep=1.5pt,fill](15) at (8,1) {};
  \node[circle,inner sep=1.5pt,fill](16) at (8,-3) {};
  
  \node[circle,inner sep=1.5pt,fill](17) at (4,-1.5) {};
  
  \node[circle,inner sep=1.5pt,fill](18) at (3,-0.666) {};
  \node[circle,inner sep=1.5pt,fill](19) at (3,-2.333) {};
  %\node (20) at (3.25, -1.315) {{\small \bfseries $\vdots$}};
  \draw[decoration={brace,mirror,raise=5pt},decorate]
  (3,-0.5) -- node[left=6pt] {$r$} (3,-2.5);
  
  \draw[thick, gray] (18) -- (17);
  \draw[thick, gray] (19) --(17);
  \draw[thick, gray] (17) --(12);
  \draw[very thick, dashed] (17) -- (11);
  \draw[very thick, dashed] (15) -- (11);
  \draw[very thick, dotted] (17) -- (13);
  \draw[very thick, dotted] (11) to [bend left] (13);
  \draw[very thick, dashdotted] (17) -- (14);
  \draw[very thick, dashdotted] (13) -- (14);
  \draw[thick, gray] (12) -- (13);
  \draw[thick, gray] (16) -- (13);
  \draw[thick, gray] (11) --(12);
  
  \node (21) at (1, -1.5) {{ $\Leftrightarrow$}};
  \node (24) at (2.125, -5) {\footnotesize Image of $3$ \emph{good} wedges};

  \node[circle,inner sep=1.5pt,fill](110) at (14,1) {};
  \node[circle,inner sep=1.5pt,fill](120) at (14,-0.666) {};
  \node[circle,inner sep=1.5pt,fill](130) at (14,-2.333) {};
  \node[circle,inner sep=1.5pt,fill](140) at (14,-4) {};
  \node[circle,inner sep=1.5pt,fill](150) at (16,1) {};
  \node[circle,inner sep=1.5pt,fill](160) at (16,-3) {};
  
  \node[circle,inner sep=1.5pt,fill](170) at (12,-1.5) {};
  
  \node[circle,inner sep=1.5pt,fill](180) at (11,-0.666) {};
  \node[circle,inner sep=1.5pt,fill](190) at (11,-2.333) {};
  %\node (200) at (11.25, -1.315) {{\small \bfseries $\vdots$}};
  \draw[decoration={brace,mirror,raise=5pt},decorate]
  (11,-0.5) -- node[left=6pt] {$r$} (11,-2.5);
    \node (230) at (13.25, -5) {{\footnotesize Image of a \emph{test} wedge}};
  
  \draw[very thick] (180) -- (170);
  \draw[thick, gray] (190) --(170);
  \draw[very thick] (170) --(120);
  \draw[thick, gray] (170) -- (110);
  \draw[thick, gray] (150) -- (110);
  \draw[thick, gray] (170) -- (130);
  \draw[thick, gray] (110) to [bend left] (130);
  \draw[thick, gray] (170) -- (140);
  \draw[thick, gray] (130) -- (140);
  \draw[thick, gray] (120) -- (130);
  \draw[thick, gray] (160) -- (130);
  \draw[thick, gray] (110) --(120);

  \node[circle,inner sep=1.5pt,fill](1100) at (22,1) {};
  \node[circle,inner sep=1.5pt,fill](1200) at (22,-0.666) {};
  \node[circle,inner sep=1.5pt,fill](1300) at (22,-2.333) {};
  \node[circle,inner sep=1.5pt,fill](1400) at (22,-4) {};
  \node[circle,inner sep=1.5pt,fill](1500) at (24,1) {};
  \node[circle,inner sep=1.5pt,fill](1600) at (24,-3) {};
  
  \node[circle,inner sep=1.5pt,fill](1700) at (20,-1.5) {};
  
  \node[circle,inner sep=1.5pt,fill](1800) at (19,-0.666) {};
  \node[circle,inner sep=1.5pt,fill](1900) at (19,-2.333) {};
 % \node (2000) at (19.25, -1.315) {{\small \bfseries $\vdots$}};
  \draw[decoration={brace,mirror,raise=5pt},decorate]
  (19,-0.5) -- node[left=6pt] {$r$} (19,-2.5);
  \node (230) at (21.25, -5) {{\footnotesize Image of a \emph{bad} wedge}};
  
  \draw[thick, gray] (1800) -- (1700);
  \draw[thick, gray] (1900) --(1700);
  \draw[thick, gray] (1700) --(1200);
  \draw[thick, gray] (1700) -- (1100);
  \draw[thick, gray] (1500) -- (1100);
  \draw[thick, gray] (1700) -- (1300);
  \draw[thick, gray] (1100) to [bend left] (1300);
  \draw[thick, gray] (1700) -- (1400);
  \draw[thick, gray] (1300) -- (1400);
  \draw[very thick] (1200) -- (1300);
  \draw[very thick] (1600) -- (1300);
  \draw[thick, gray] (1100) --(1200);

\end{tikzpicture}

%% file: figures/graphic_algo.tex
\begin{tikzpicture}[-,thick,xscale=0.5,yscale=0.4,vc node/.style={inner sep=2.5pt,fill},nvc node/.style={circle,inner sep=1.5pt,fill}]
  \node[vc node, label=left:{$x$}] (1) at (0,0) {};
  \node[vc node, label=left:{$w$}] (2) at (0,2) {};
  \node[vc node, label=left:{$v$}] (3) at (0,4) {};
  \node[vc node, label=left:{$u$}] (4) at (0,6) {};
  
  \node[nvc node, label=right:{$f$}] (5) at (2,0) {};
  \node[nvc node, label=right:{$e$}] (6) at (2,1.5) {};
  \node[nvc node, label=right:{$d$}] (7) at (2,2.5) {};
  \node[nvc node, label=right:{$c$}] (8) at (2,4) {};
  \node[nvc node, label=right:{$b$}] (9) at (2,5) {};
  \node[nvc node, label=right:{$a$}] (10) at (2,6) {};
  
  \node (42) at (1,-1) {\large $H$};
  
  \draw (1) --(2); \draw (3) --(4);\draw (1) --(5);\draw (3) --(9);
  \draw (2) --(6); \draw (7) --(2);\draw (3) --(8);\draw (4) --(9);\draw (4) --(10);
  
  \node[vc node, label=left:{$w$}] (2) at (5,0) {};
  \node[vc node, label=left:{$vx$}] (3) at (5,3) {};
  \node[vc node, label=left:{$u$}] (4) at (5,6) {};
  
  \node[nvc node, label=right:{$f$}] (5) at (7,2.5) {};
  \node[nvc node, label=right:{$e$}] (6) at (7,0) {};
  \node[nvc node, label=right:{$d$}] (7) at (7,1) {};
  \node[nvc node, label=right:{$c$}] (8) at (7,3.5) {};
  \node[nvc node, label=right:{$b$}] (9) at (7,5) {};
  \node[nvc node, label=right:{$a$}] (10) at (7,6) {};
  
  \node (42) at (6,-1) {\large $H/\rho_C$};
  
  \draw (3) --(2);\draw (3) --(4);\draw (3) --(5);\draw (3) --(9);
  \draw (2) --(6); \draw (7) --(2);\draw (3) --(8);\draw (4) --(9);\draw (4) --(10);
  
  \node[vc node, label=left:{$w$}] (2) at (10,0) {};
  \node[vc node, label=left:{$vx$}] (3) at (10,3) {};
  \node[vc node, label=left:{$u$}] (4) at (10,6) {};
  
  \node[nvc node, label=right:{$b$}] (5) at (12,6) {};
  \node[nvc node, label=right:{$~$}] (6) at (12,3) {};
  \node[inner sep=2.5pt, label=right:{$acd$}] (60) at (12,3.1) {};
  \node[nvc node, label=right:{$ef$}] (7) at (12,0) {};
  
  \node (42) at (11,-1) {\large $H/\rho_1$};
  
  \draw (3) --(2);\draw (3) --(4); \draw (7) --(2);\draw (2) --(6);
  \draw (3) --(5);\draw (3) --(6); \draw (3) --(7);\draw (4) --(5);\draw (6) --(4);
  
  \node[vc node, label=left:{$w$}] (2) at (15,0) {};
  \node[vc node, label=left:{$vx$}] (3) at (15,3) {};
  \node[vc node, label=left:{$u$}] (4) at (15,6) {};
  
  \node[nvc node, label=right:{$b$}] (5) at (17,6) {};
  \node[nvc node, label=right:{$~$}] (6) at (17,3) {};
  \node[inner sep=2.5pt, label=right:{$aef$}] (60) at (17,3.05) {};
  \node[nvc node, label=right:{$cd$}] (7) at (17,0) {};
  
  \node (42) at (16,-1) {\large $H/\rho_2$};
  
  \draw (3) --(2);\draw (3) --(4); \draw (7) --(2);\draw (2) --(6);
  \draw (3) --(5);\draw (3) --(6); \draw (3) --(7);\draw (4) --(5);\draw (6) --(4);
  
  \node[vc node, label=left:{$w$}] (2) at (20,0) {};
  \node[vc node, label=left:{$vx$}] (3) at (20,3) {};
  \node[vc node, label=left:{$u$}] (4) at (20,6) {};
  
  \node[nvc node, label=right:{$ac$}] (5) at (22,6) {};
  \node[nvc node, label=right:{$bd$}] (6) at (22,3) {};
  \node[nvc node, label=right:{$ef$}] (7) at (22,0) {};
  
  \node (42) at (21,-1) {\large $H/\rho_3$};
  
  \draw (3) --(2);\draw (3) --(4); \draw (7) --(2);\draw (2) --(6);
  \draw (3) --(5);\draw (3) --(6); \draw (3) --(7);\draw (4) --(5);\draw (6) --(4);

\end{tikzpicture}

%% file: figures/graphic_gzero.tex
\begin{tikzpicture}[-, scale = 0.8, yscale=0.6]
  \node[circle,inner sep=1.5pt,fill](1) at (-2,1) {};
  \node[circle,inner sep=1.5pt,fill](2) at (-2,-0.666) {};
  \node[circle,inner sep=1.5pt,fill](3) at (-2,-2.333) {};
  \node[circle,inner sep=1.5pt,fill](4) at (-2,-4) {};
  \node[circle,inner sep=1.5pt,fill](5) at (0,1) {};
  \node[circle,inner sep=1.5pt,fill](6) at (0,0) {};
  \node[circle,inner sep=1.5pt,fill](7) at (0,-1) {};
  \node[circle,inner sep=1.5pt,fill](8) at (0,-2) {};
  \node[circle,inner sep=1.5pt,fill](9) at (0,-3) {};
  \node[circle,inner sep=1.5pt,fill](10) at (0,-4) {};
  
  \draw[very thick, dashed] (1)--(5);
  \draw[thick, gray] (1)--(6);
  \draw[thick, gray] (1)--(7);
  \draw[thick, gray] (2)--(6);
  \draw[thick, gray] (2)--(8);
  \draw[very thick, dotted] (3)--(7);
  \draw[thick, gray] (3)--(8);
  \draw[thick, gray] (3)--(9);
  \draw[thick, gray] (3)--(10);
  \draw[very thick, dashdotted] (4) -- (10);
  
  \node[circle,inner sep=1.5pt,fill](11) at (8,1) {};
  \node[circle,inner sep=1.5pt,fill](12) at (8,-0.666) {};
  \node[circle,inner sep=1.5pt,fill](13) at (8,-2.333) {};
  \node[circle,inner sep=1.5pt,fill](14) at (8,-4) {};
  \node[circle,inner sep=1.5pt,fill](15) at (10,1) {};
  \node[circle,inner sep=1.5pt,fill](16) at (10,-3) {};
  
  \node[circle,inner sep=1.5pt,fill](17) at (6,-1.5) {};
  
  \node[circle,inner sep=1.5pt,fill](18) at (5,-0.666) {};
  \node[circle,inner sep=1.5pt,fill](19) at (5,-2.333) {};

  \draw[very thick] (18) -- (19);
  \draw[very thick] (19) --(17);
  \draw[thick, gray] (17) --(12);
  \draw[very thick, dashed] (17) -- (11);
  \draw[very thick, dashed] (15) -- (11);
  \draw[very thick, dotted] (17) -- (13);
  \draw[very thick, dotted] (11) to [bend left] (13);
  \draw[very thick, dashdotted] (17) -- (14);
  \draw[very thick, dashdotted] (13) -- (14);
  \draw[thick, gray] (12) -- (13);
  \draw[thick, gray] (16) -- (13);
  \draw[thick, gray] (11) --(12);
  
  \node (21) at (2.5, -1.5) {{\huge $\Rightarrow$}};
  \node (22) at (-1, 2) {{\Large $G$}};
  \node (23) at (9, 2) {{\Large $G'$}};
  
  \node (24) at (6, -2.25) {{$0$}};
  \node (25) at (5, -3) {{$1$}};
  \node (26) at (5, 0) {{$2$}};

\end{tikzpicture}

%% file: ms.bbl
\begin{thebibliography}{10}

\bibitem{DBLP:journals/corr/CaiF16}
Jin{-}Yi Cai and Zhiguo Fu.
\newblock Holographic algorithm with matchgates is universal for planar \#{CSP}
  over boolean domain.
\newblock {\em CoRR}, abs/1603.07046, 2016.
\newblock URL: \url{http://arxiv.org/abs/1603.07046}.

\bibitem{Cai.Hemachandra1990}
Jin{-}Yi Cai and Lane~A. Hemachandra.
\newblock On the power of parity polynomial time.
\newblock {\em Mathematical Systems Theory}, 23(2):95--106, 1990.
\newblock \href {http://dx.doi.org/10.1007/BF02090768}
  {\path{doi:10.1007/BF02090768}}.

\bibitem{DBLP:journals/jcss/CaiL11}
Jin{-}yi Cai and Pinyan Lu.
\newblock Holographic algorithms: From art to science.
\newblock {\em J. Comput. Syst. Sci.}, 77(1):41--61, 2011.
\newblock \href {http://dx.doi.org/10.1016/j.jcss.2010.06.005}
  {\path{doi:10.1016/j.jcss.2010.06.005}}.

\bibitem{Cai.Lu2008}
Jin-Yi Cai, Pinyan Lu, and Mingji Xia.
\newblock Holographic algorithms by {F}ibonacci gates and holographic
  reductions for hardness.
\newblock In {\em Proceedings of the 49th Annual Symposium on Foundations of
  Computer Science, FOCS}, pages 644--653, 2008.
\newblock \href {http://dx.doi.org/10.1109/FOCS.2008.34}
  {\path{doi:10.1109/FOCS.2008.34}}.

\bibitem{Cai.Lu.Xia2011}
Jin-Yi Cai, Pinyan Lu, and Mingji Xia.
\newblock Computational complexity of holant problems.
\newblock {\em SIAM Journal on Computing}, 40(4):1101--1132, 2011.
\newblock \href {http://dx.doi.org/10.1137/100814585}
  {\path{doi:10.1137/100814585}}.

\bibitem{DBLP:journals/siamcomp/CaiLX17}
Jin{-}Yi Cai, Pinyan Lu, and Mingji Xia.
\newblock Holographic algorithms with matchgates capture precisely tractable
  planar {\#}{CSP}.
\newblock {\em {SIAM} J. Comput.}, 46(3):853--889, 2017.
\newblock \href {http://dx.doi.org/10.1137/16M1073984}
  {\path{doi:10.1137/16M1073984}}.

\bibitem{Chen.Chor2005}
Jianer Chen, Benny Chor, Mike Fellows, Xiuzhen Huang, David~W. Juedes, Iyad~A.
  Kanj, and Ge~Xia.
\newblock Tight lower bounds for certain parameterized {NP}-hard problems.
\newblock {\em Information and Computation}, 201(2):216--231, 2005.
\newblock \href {http://dx.doi.org/10.1016/j.ic.2005.05.001}
  {\path{doi:10.1016/j.ic.2005.05.001}}.

\bibitem{Chen.Thurley2008}
Yijia Chen, Marc Thurley, and Mark Weyer.
\newblock Understanding the complexity of induced subgraph isomorphisms.
\newblock In {\em Proceedings of the 35th International Colloquium on Automata,
  Languages and Programming, ICALP}, pages 587--596, 2008.
\newblock \href {http://dx.doi.org/10.1007/978-3-540-70575-8_48}
  {\path{doi:10.1007/978-3-540-70575-8_48}}.

\bibitem{chudnovsky2006strong}
Maria Chudnovsky, Neil Robertson, Paul Seymour, and Robin Thomas.
\newblock The strong perfect graph theorem.
\newblock {\em Annals of Mathematics}, 164:51--229, 2006.
\newblock \href {http://dx.doi.org/10.4007/annals.2006.164.51}
  {\path{doi:10.4007/annals.2006.164.51}}.

\bibitem{DBLP:journals/iandc/CreignouH96}
Nadia Creignou and Miki Hermann.
\newblock Complexity of generalized satisfiability counting problems.
\newblock {\em Inf. Comput.}, 125(1):1--12, 1996.
\newblock \href {http://dx.doi.org/10.1006/inco.1996.0016}
  {\path{doi:10.1006/inco.1996.0016}}.

\bibitem{Curticape2013}
Radu Curticapean.
\newblock Counting matchings of size $k$ is \#{W}[1]-hard.
\newblock In {\em Proceedings of the 40th International Colloquium on Automata,
  Languages and Programming, ICALP}, pages 352--363, 2013.
\newblock \href {http://dx.doi.org/10.1007/978-3-642-39206-1_30}
  {\path{doi:10.1007/978-3-642-39206-1_30}}.

\bibitem{Curticapean.PhD}
Radu Curticapean.
\newblock {\em The simple, little and slow things count: On parameterized
  counting complexity}.
\newblock PhD thesis, Saarland University, August 2015.

\bibitem{DBLP:conf/stoc/CurticapeanDM17}
Radu Curticapean, Holger Dell, and D{\'{a}}niel Marx.
\newblock Homomorphisms are a good basis for counting small subgraphs.
\newblock In {\em Proceedings of the 49th ACM Symposium on Theory of Computing,
  STOC}, pages 210--223, 2017.
\newblock \href {http://dx.doi.org/10.1145/3055399.3055502}
  {\path{doi:10.1145/3055399.3055502}}.

\bibitem{curticapean2014complexity}
Radu Curticapean and D\'aniel Marx.
\newblock Complexity of counting subgraphs: Only the boundedness of the
  vertex-cover number counts.
\newblock In {\em Proceedings of the 55th Annual Symposium on Foundations of
  Computer Science, FOCS}, pages 130--139. IEEE, 2014.
\newblock \href {http://dx.doi.org/10.1109/FOCS.2014.22}
  {\path{doi:10.1109/FOCS.2014.22}}.

\bibitem{DBLP:conf/focs/CurticapeanX15}
Radu Curticapean and Mingji Xia.
\newblock Parameterizing the permanent: Genus, apices, minors, evaluation mod
  $2^k$.
\newblock In {\em Proceedings of the 56th Annual Symposium on Foundations of
  Computer Science, FOCS}, pages 994--1009, 2015.
\newblock \href {http://dx.doi.org/10.1109/FOCS.2015.65}
  {\path{doi:10.1109/FOCS.2015.65}}.

\bibitem{Dagum.Luby1992}
Paul Dagum and Michael Luby.
\newblock Approximating the permanent of graphs with large factors.
\newblock {\em Theoretical Computer Science}, 102(2):283--305, 1992.
\newblock \href {http://dx.doi.org/10.1016/0304-3975(92)90234-7}
  {\path{doi:10.1016/0304-3975(92)90234-7}}.

\bibitem{Dalmau.Jonsson2004}
V\'ictor Dalmau and Peter Jonsson.
\newblock The complexity of counting homomorphisms seen from the other side.
\newblock {\em Theoretical Computer Science}, 329(1-3):315--323, 2004.
\newblock \href {http://dx.doi.org/10.1016/j.tcs.2004.08.008}
  {\path{doi:10.1016/j.tcs.2004.08.008}}.

\bibitem{Dell.Husfeldt2014}
Holger Dell, Thore Husfeldt, D{\'{a}}niel Marx, Nina Taslaman, and Martin
  Wahlen.
\newblock Exponential time complexity of the permanent and the {T}utte
  polynomial.
\newblock {\em ACM Transactions on Algorithms}, 10(4):21, 2014.
\newblock \href {http://dx.doi.org/10.1145/2635812}
  {\path{doi:10.1145/2635812}}.

\bibitem{Flum.Grohe2004}
J\"org Flum and Martin Grohe.
\newblock The parameterized complexity of counting problems.
\newblock {\em SIAM Journal of Computing}, 33(4):892--922, 2004.
\newblock \href {http://dx.doi.org/10.1137/S0097539703427203}
  {\path{doi:10.1137/S0097539703427203}}.

\bibitem{FlumGrohe}
J{\"o}rg Flum and Martin Grohe.
\newblock {\em Parameterized complexity theory}.
\newblock Springer, 2006.

\bibitem{Frick2004}
Markus Frick.
\newblock Generalized model-checking over locally tree-decomposable classes.
\newblock {\em Theoretical Computer Science}, 37(1):157--191, 2004.
\newblock \href {http://dx.doi.org/10.1007/s00224-003-1111-9}
  {\path{doi:10.1007/s00224-003-1111-9}}.

\bibitem{Grohe2007}
Martin Grohe.
\newblock The complexity of homomorphism and constraint satisfaction problems
  seen from the other side.
\newblock {\em Journal of the ACM}, 54(1):1, 2007.
\newblock \href {http://dx.doi.org/10.1145/1206035.1206036}
  {\path{doi:10.1145/1206035.1206036}}.

\bibitem{Grohe.Schwentick2001}
Martin Grohe, Thomas Schwentick, and Luc Segoufin.
\newblock When is the evaluation of conjunctive queries tractable?
\newblock In {\em Proceedings of the 33rd ACM Symposium on Theory of Computing,
  STOC}, pages 657--666, 2001.
\newblock \href {http://dx.doi.org/10.1145/380752.380867}
  {\path{doi:10.1145/380752.380867}}.

\bibitem{DBLP:journals/dam/Guruswami99}
Venkatesan Guruswami.
\newblock Maximum cut on line and total graphs.
\newblock {\em Discrete Applied Mathematics}, 92(2-3):217--221, 1999.
\newblock \href {http://dx.doi.org/10.1016/S0166-218X(99)00056-6}
  {\path{doi:10.1016/S0166-218X(99)00056-6}}.

\bibitem{harary2004graph}
Frank Harary.
\newblock {\em Graph theory}.
\newblock Addison-Wesley, Reading, 1969.

\bibitem{Jerrum1987}
Mark Jerrum.
\newblock Two-dimensional monomer-dimer systems are computationally
  intractable.
\newblock {\em Journal of Statistical Physics}, 48(1-2):121--134, 1987.
\newblock \href {http://dx.doi.org/10.1007/BF01010403}
  {\path{doi:10.1007/BF01010403}}.

\bibitem{Jerrum.Sinclair2004}
Mark Jerrum, Alistair Sinclair, and Eric Vigoda.
\newblock A polynomial-time approximation algorithm for the permanent of a
  matrix with nonnegative entries.
\newblock {\em Journal of the ACM}, 51(4):671--697, 2004.
\newblock \href {http://dx.doi.org/10.1145/1008731.1008738}
  {\path{doi:10.1145/1008731.1008738}}.

\bibitem{Kasteleyn1967}
Pieter~W. Kasteleyn.
\newblock Graph theory and crystal physics.
\newblock In {\em Graph Theory and Theoretical Physics}, pages 43--110.
  Academic Press, 1967.

\bibitem{DBLP:journals/jacm/Lehot74}
Philippe G.~H. Lehot.
\newblock An optimal algorithm to detect a line graph and output its root
  graph.
\newblock {\em Journal of the ACM}, 21(4):569--575, 1974.
\newblock \href {http://dx.doi.org/10.1145/321850.321853}
  {\path{doi:10.1145/321850.321853}}.

\bibitem{lovaszbook}
L{\'a}szl{\'o} Lov{\'a}sz.
\newblock {\em Large networks and graph limits}, volume~60.
\newblock American Mathematical Society Providence, 2012.

\bibitem{Lozin200574}
Vadim~V. Lozin and Raffaele Mosca.
\newblock Independent sets in extensions of $2{K}_2$-free graphs.
\newblock {\em Discrete Applied Mathematics}, 146(1):74--80, 2005.
\newblock \href {http://dx.doi.org/10.1016/j.dam.2004.07.006}
  {\path{doi:10.1016/j.dam.2004.07.006}}.

\bibitem{DBLP:journals/dam/Meeks16}
Kitty Meeks.
\newblock The challenges of unbounded treewidth in parameterised subgraph
  counting problems.
\newblock {\em Discrete Applied Mathematics}, 198:170--194, 2016.
\newblock \href {http://dx.doi.org/10.1016/j.dam.2015.06.019}
  {\path{doi:10.1016/j.dam.2015.06.019}}.

\bibitem{Roth17}
Marc Roth.
\newblock Counting restricted homomorphisms via {M\"{o}bius} inversion over
  matroid lattices.
\newblock In {\em 25th Annual European Symposium on Algorithms, {ESA} 2017,
  September 4-6, 2017, Vienna, Austria}, pages 63:1--63:14, 2017.
\newblock \href {http://dx.doi.org/10.4230/LIPIcs.ESA.2017.63}
  {\path{doi:10.4230/LIPIcs.ESA.2017.63}}.

\bibitem{zbMATH03693322}
Najiba Sbihi.
\newblock Algorithme de recherche d'un stable de cardinalite maximum dans un
  graphe sans etoile.
\newblock {\em Discrete Mathematics}, 29:53--76, 1980.
\newblock \href {http://dx.doi.org/10.1016/0012-365X(90)90287-R}
  {\path{doi:10.1016/0012-365X(90)90287-R}}.

\bibitem{Temperley.Fisher1961}
Harold N.~V. Temperley and Michael~E. Fisher.
\newblock Dimer problem in statistical mechanics - an exact result.
\newblock {\em Philosophical Magazine}, 6(68):1478--6435, 1961.
\newblock \href {http://dx.doi.org/10.1080/14786436108243366}
  {\path{doi:10.1080/14786436108243366}}.

\bibitem{Vadh2001}
Salil~P. Vadhan.
\newblock The complexity of counting in sparse, regular, and planar graphs.
\newblock {\em SIAM Journal of Computing}, 31(2):398--427, 2001.
\newblock \href {http://dx.doi.org/10.1137/S0097539797321602}
  {\path{doi:10.1137/S0097539797321602}}.

\bibitem{Valiant1979a}
Leslie~G. Valiant.
\newblock The complexity of computing the permanent.
\newblock {\em Theoretical Computer Science}, 8(2):189--201, 1979.
\newblock \href {http://dx.doi.org/10.1016/0304-3975(79)90044-6}
  {\path{doi:10.1016/0304-3975(79)90044-6}}.

\bibitem{Valiant2008}
Leslie~G. Valiant.
\newblock Holographic algorithms.
\newblock {\em SIAM Journal of Computing}, 37(5):1565--1594, 2008.
\newblock \href {http://dx.doi.org/10.1137/070682575}
  {\path{doi:10.1137/070682575}}.

\bibitem{Sol94}
\v{L}. \v{S}olt\'es.
\newblock Forbidden induced subgraphs for line graphs.
\newblock {\em Discrete Mathematics}, 132(1):391--394, 1994.
\newblock \href {http://dx.doi.org/10.1016/0012-365X(92)00577-E}
  {\path{doi:10.1016/0012-365X(92)00577-E}}.

\bibitem{williams2013finding}
Virginia~Vassilevska Williams and Ryan Williams.
\newblock Finding, minimizing, and counting weighted subgraphs.
\newblock {\em SIAM Journal on Computing}, 42(3):831--854, 2013.
\newblock \href {http://dx.doi.org/10.1137/09076619X}
  {\path{doi:10.1137/09076619X}}.

\bibitem{Xia.Zhang2007}
Mingji Xia, Peng Zhang, and Wenbo Zhao.
\newblock Computational complexity of counting problems on 3-regular planar
  graphs.
\newblock {\em Theoretical Computer Science}, 384(1):111 -- 125, 2007.
\newblock Theory and Applications of Models of Computation.
\newblock \href {http://dx.doi.org/10.1016/j.tcs.2007.05.023}
  {\path{doi:10.1016/j.tcs.2007.05.023}}.

\bibitem{zhang2004disjoint}
Xiao-Dong Zhang and Stanislaw Bylka.
\newblock Disjoint triangles of a cubic line graph.
\newblock {\em Graphs and Combinatorics}, 20(2):275--280, 2004.
\newblock \href {http://dx.doi.org/10.1007/s00373-004-0551-6}
  {\path{doi:10.1007/s00373-004-0551-6}}.

\end{thebibliography}
